\documentclass[11pt]{article}
\usepackage{amsmath}
\usepackage{amssymb}
\usepackage{fullpage}
\usepackage{amsthm}
\usepackage{epsf}
\usepackage{graphicx}
\usepackage{color}
\usepackage{enumerate}

\newcommand{\ignore}[1]{}

\theoremstyle{plain}

\newtheorem{ques}{Question}
\newtheorem{theorem}{Theorem}
\newtheorem{lemma}[theorem]{Lemma}

\newtheorem{prop}[theorem]{Proposition}
\newtheorem{corollary}[theorem]{Corollary}
\newtheorem{definition}{Definition}

\theoremstyle{definition}
\newtheorem{remark}{Remark}

\renewenvironment{proof}{\noindent {\bf \em Proof\/}.\enspace}
{\hfill $\blacksquare{}$ \vspace{12pt}}

\def\av{\mathop{\bf E}\limits}

\renewcommand{\le}{\leqslant}
\renewcommand{\ge}{\geqslant}

\newcommand{\eps}{\varepsilon}
\renewcommand{\epsilon}{\varepsilon}

\newcommand{\vnote}[1]{}
\newcommand{\anote}[1]{}
\newcommand{\arnote}[1]{}

\newcommand{\F}{\mathbb{F}}

\newcommand{\prob}[1]{\mbox{\rm Pr}\left[#1\right]}

\renewcommand{\a}{\alpha}
\renewcommand{\b}{\beta}
\newcommand{\g}{\gamma}

\newcommand{\rs}{{\rm RS}}
\newcommand{\grm}{{\rm RM}}

\newcommand{\myvec}[1]{\vec{#1}}

\newcommand{\vzero}{\vec{0}}
\newcommand{\vs}{\vec{s}}

\newcommand{\vx}{\vec{x}}
\newcommand{\vxhat}{\vec{\hat{x}}}
\newcommand{\vy}{\vec{y}}
\newcommand{\vz}{\vec{z}}
\newcommand{\vr}{\vec{r}}

\newcommand{\vc}{\vec{c}}
\newcommand{\va}{\vec{a}}

\newcommand{\ve}{\vec{e}}

\newcommand{\ttds}{tolerant tester}
\newcommand{\ttq}{tolerant tester}
\newcommand{\ltds}{local tester}
\newcommand{\edds}{error detector}

\newcommand{\cc}{\mathrm{CC}}
\newcommand{\poly}[1]{\mathrm{poly}(#1)}

\newcommand{\wt}{\textsc{wt}}

\title{{\bf Data Stream Algorithms for Codeword Testing}\footnote{Research supported by NSF CAREER Award CCF-0844796.}}

\author{Atri Rudra~~~ Steve Uurtamo}

\date{Department of Computer Science and Engineering,\\
University at Buffalo, The State University of New York,\\
Buffalo, NY, 14620.\\ {\tt \{atri,uurtamo\}@buffalo.edu}}

\begin{document}
\maketitle

\begin{abstract}

Motivated by applications in storage systems and property testing, we study
data stream algorithms for local testing and tolerant testing of
codes. Ideally, we would like to know whether there exist asymptotically good
codes that can be local/tolerant tested with one-pass, poly-log space data
stream algorithms.

We show that for the error detection problem (and hence, the local testing
problem), there exists a one-pass, log-space data stream algorithm for a
broad class of asymptotically good codes, including the Reed-Solomon (RS) code
and expander codes.
In our technically more involved result, we give a one-pass, 
$O(e\log^2{n})$-space algorithm for
RS (and related) codes with dimension $k$ and block length $n$
that can distinguish between the cases when the Hamming distance between
the received word and the code is at most $e$ and at least $a\cdot e$ for
some absolute constant $a>1$. For RS codes with random errors, 
we can obtain $e\le O(n/k)$.
For  folded RS codes, we obtain
similar results for worst-case errors as long as $e\le (n/k)^{1-\eps}$ for
any constant $\eps>0$. 
These results follow by reducing the tolerant testing problem to the error
detection problem using results from group testing and the list decodability
of the code. We also show that using our techniques, the space
requirement and the upper bound of $e\le O(n/k)$ cannot be improved by more
than logarithmic factors.


\end{abstract}

\section{Introduction}

In this work, we consider data stream algorithms for local testing and tolerant
testing of error-correcting codes. The local testing
problem for a code $C\subseteq \Sigma^n$ is the following: given a received
word $\vy\in\Sigma^n$, we need to figure out if $\vy\in C$ or if $\vy$ differs
from every codeword in $C$ in at least $0<e\le n$ positions (i.e. the
Hamming distance of $\vy$ from every $\vc\in C$, denoted by $\Delta(\vy,\vc)$,
is at least $e$). If $e=1$, then this is the error-detection problem.
In the tolerant testing problem, given $\vy$, we need to decide whether $\vy$
is at a distance at most $e_1$ from some codeword or if it has distance at least
$e_2>e_1$ from every codeword in $C$. Ideally, we would like to answer the
following (see Section~\ref{sec:prelims} for definitions related to codes):
\begin{ques}
\label{ques:main}
Do there exist asymptotically good codes that can be (tolerant) tested
by one-pass, poly-log space data stream algorithms?
\end{ques}

To the best of our knowledge, ours is the first work that considers this
natural problem.
We begin with the motivation for our work.

\vspace*{2mm}
\noindent
\textbf{Property Testing.} Local testing of codes has been extensively studied under the
stricter requirements of property testing. Under
the property testing requirements, one needs to solve the local
testing problem by (ideally) only accessing  a constant number of
positions in $\vy$. Codes that can be locally tested with a constant number
of queries have been instrumental in the development of the PCP machinery,
starting with the original proof of the PCP theorem~\cite{ALMSS98,AS98}.
The current record is due to Dinur, who presents codes that
have inverse poly-log rate, linear distance and can be locally tested
with a constant number of queries~\cite{dinur}.

General lower bounds on local testing of codes, however, have been scarce.
(See, e.g. the recent paper~\cite{redundant}). In particular, it is not known
if there are asymptotically good codes that can be locally tested with a
constant number of queries.  The question remains open even if one considers
the harder task of tolerant testing with a constant number of
\textit{non-adaptive} queries \cite{GR-tolerant}. 

It is not too hard to see that a non-adaptive tolerant tester
that makes a constant number of queries gives a single pass, log-space
data stream algorithm for tolerant testing. (See
Section~\ref{app:prop-test} for a proof.) Thus, if one could prove that
any one-pass data stream algorithm for local/tolerant testing of
asymptotically good codes requires $\omega(\log{n})$ space, then they will
have answered the question in the negative (at least for non-adaptive
queries). This could present a new approach to attack the
question of local/tolerant testing with a constant number of queries.

Next, we discuss the implications of  \textit{positive} results
for local/tolerant testing.

\vspace*{2mm}
\noindent
\textbf{Applications in Storage Systems.}  Codes are used in current day
storage systems such as optical storage (CDs and DVDs), RAID
(\cite{chen94raid}) and ECC memory (\cite{ecc-mem}). Storage systems, up
until recently, used simple codes such as the parity code and
checksums,
which have (trivial) data stream algorithms for error detection.
However, the parity code cannot detect even two errors. With the explosion
in the amount of data that needs to be stored, errors are becoming more
frequent. This situation will become worse as more data gets stored on
disks~\cite{disk-cacm}. Thus, we need to use codes that can handle more errors.

Reed-Solomon (RS) codes, which are used widely in storage systems (e.g. in
CDs and DVDs and more recently in RAID), are well-known to have good
error correcting capabilities. However, the conventional error detection
for RS codes is not space or pass efficient. Thus, a natural question to
ask is if one can design a data stream algorithm to perform error
detection for RS codes.

It would be remiss of us not to point out that unlike a typical application
of a data stream algorithm where $n$ is very large, in real life deployments
of RS codes, $n$ is relatively small.  However, if one needs to implement
the error detection algorithm in controllers on disks then it would be
advantageous to use a data stream algorithm so that it is feasible to
perform error detection with every read.  Another way to use error
detection is in \textit{data scrubbing} \cite{disk-cacm}. In this scenario, 
during idle time or low activity periods, error detection is run on
the entire disk to catch errors. In addition,
the single pass requirement means that we will probe each bit on a disk
only once, which is good for longevity of data. Finally, it would be helpful
to complement an error-detection algorithm with a data stream
algorithm that could also \textit{locate} the errors.

It is also plausible that the efficiency of the data stream algorithms will
make it feasible to use RS codes of block length (and alphabet size)
considerably larger than the ones currently used in practice.

Before we delve into the description of our results, we would like to point out
a few things. First, for the storage application, designing
algorithms for a widely used code such as the RS code will be more valuable
than answering Question~\ref{ques:main} in the affirmative via some new code.
Second, it is known that for local testing of a RS code of dimension $k$, at
least $k$ queries need to be made in the property testing
world\footnote{This follows from the fact that the ``dual" code has
distance $k+1$.}. 
However, this does not rule out the
possibility of a local/tolerant tester in the data stream world. Finally,
for storage systems, solving the tolerant testing problem even for
large constants $e_1$ and $e_2$ would be interesting.

\vspace*{2mm}
\noindent
\textbf{Our Results.}
We give a one-pass, poly-log space, randomized algorithm to perform error
detection for a broad class of asymptotically good codes such as
Reed-Solomon (RS) and expander codes. As a complementary result, we also
show that deterministic data stream algorithms (even with multiple passes)
require linear space for such codes.  Thus, for local testing
we answer Question~\ref{ques:main} in the affirmative. This should be
contrasted with the situation in property testing, where it
is known that for both asymptotically good RS and expander codes, a linear
number of queries is required. (The lower bound for RS codes was
discussed in the paragraph above and the result for expander codes
follows from~\cite{cnf-hard}.)

It turns out that using existing results for tolerant testing of Reed-Muller
codes over large alphabets~\cite{GR-tolerant}, one can answer
Question~\ref{ques:main} in the affirmative for tolerant testing, though with $O(n^{\eps})$
space for any constant $\eps>0$. (See Section~\ref{app:rm} for more
details.)

Given the practical importance of RS codes and given the fact that local
testing for RS codes with data stream constraints is possible, for the rest of
the paper we focus mostly on tolerant testing of RS and related codes.
We first remark that a naive tolerant testing algorithm for RS codes that
can be implemented in $O(e\log{n})$ space is to go through all the
$\sum_{i=1}^e\binom{n}{i}$ possible error locations $S$ and check if the
received word projected outside of $S$ belongs to the corresponding
projected down RS code. (This works as long as $e\le n-k$, which
is true w.l.o.g. since $n-k$ is the covering radius of a RS code of
dimension $k$ and block length $n$.)
Using our error detection algorithm for RS codes, this can be implemented
as a one-pass $O(e\log{n})$-space data stream algorithm. Unfortunately,
the running time of this algorithm is prohibitive, even for moderate
values of $e$.

In this paper, we match the parameters above to within a log factor but with a 
(small) polynomial running
time for values of $e$ much larger than a constant.
In particular, we present a one-pass, $O(e\log^2{n})$-space, polynomial time 
randomized algorithm for a RS code $C$ with dimension $k$ and block length $n$
that can distinguish between the cases when the Hamming distance between
$\vy$ and $C$ is at most $e$ and at least $a\cdot e$ (for some constant
$a>1$). This reduction works when $e(e+k)\le O(n)$. If we are dealing with
random errors, then we can obtain $ek\le O(n)$.
Using known results on list decodability of folded RS
codes~\cite{GR-capacity}, we obtain similar results for worst
case errors for $e\le (n/k)^{1-\eps}$ for any constant $\eps>0$. As a
byproduct, our algorithms also locate the errors (if the
number of errors is bounded by $e$), which is desirable for a storage
application.
We also show that using our techniques, the space requirement and
the upper bound of $e\le O(n/k)$ cannot be improved by more than 
logarithmic factors.

Ideally, we would like our data stream algorithms to spend poly-log
time per input position. However, in this paper we will tolerate
polynomial time algorithms.
In particular,  naive implementations of the tolerant testing algorithms take
$\tilde{O}(n^2)$ time. We also show that at the expense of slightly worse
parameters, we can achieve a running time of $\tilde{O}(ne)$ for certain
RS codes.

\vspace*{2mm}
\noindent
\textbf{Our Techniques.} 
It is well-known that error detection for any linear code can be done
by checking if the product of the received word with the parity check
matrix is the all zeros vector. We turn this into a one-pass
low space data stream algorithm using the well-known finger
printing method.
The only difference from the usual fingerprinting method, where one uses any
large enough field, is that we need to use a large enough extension field of
the finite field over which the code is defined. To show the necessity of
randomness, we use the well-known fooling set method from communication
complexity~\cite{cc-book}.
However, unlike the usual application of two-party communication complexity in
data stream algorithms, where the stream is broken up into two fixed
portions, in our case we need to be careful about how we divide up the input.  Details can be found in Section~\ref{sec:ed}.

We now move on to our tolerant testing algorithm. 
We begin with the connection to group testing.
Let $\vc$ be the closest codeword to $\vy$ and let $\vx\in\{0,1\}^n$
denote the binary vector where $x_i=1$ iff $y_i\neq c_i$. Now assume
we could access $\vx$ in the following manner: pick a subset $Q\subseteq [n]$
and check if $\vx_Q=\vzero$ or not (where $\vx_Q$ denotes $\vx$ projected
onto indices in $Q$). Then can we come up with a clever way
of non-adaptively choosing the tests such that at the end we know whether
$\wt(\vx)\le e$ or not? It turns out that we can use group testing to
construct such an algorithm. In fact, using $e$-disjunct
matrices (cf.~\cite{test-book}), we can design non-adaptive tests such
that given the answers to the tests one could compute $\vx$ if
$\wt(\vx)\le e$, else determine that $\wt(\vx)>e$. (A natural question
is how many tests do $e$-disjunct matrices require: we will come back to
this question in a bit.)
This seems to let us test whether $\vy$ is within a Hamming distance of
$e$ from some codeword or not. We would like to point out that the
above is essentially reducing one instance of the tolerant testing problem
to multiple instances of error-detection.

Thus, all we need to do is come up with a way to implement the tests to
$\vx$. A natural option, which we take, is that for any test $Q\subseteq [n]$,
we check if $\vy_Q\in \rs_Q[k]$, where $\rs_Q[k]$ is the RS code (of dimension $k$) projected
onto $Q$.
This immediately puts in one restriction: we will need $|Q|\ge k$ (as
otherwise every test will return a positive answer). However, there is another
subtle issue that makes our analysis more complicated-- we do not necessarily
have that $\vx_Q=\vzero$ iff $\vy_Q\in \rs_Q[k]$. While it is true that
$\vx_Q=0$ implies $\vy_Q\in RS_Q[k]$, the other direction is not true.
The latter is possible only if $\vy$ agrees with some codeword $\vc'\neq \vc$
in the positions indexed by $Q$. Now if $s$ is the size of the smallest
test and it is the case that the only codeword that agrees with
$\vy$ in at least $s$ positions is $\vc$, then we'll be done. We show that this
latter condition is true for RS codes if $s\ge e+k$ for worst-case errors
or with high probability if $s\ge 4k$ for random errors.

It is now perhaps not surprising that the list decodability of the code
plays a role in our general result for worst-case errors.
Assume that the code $C$ under consideration is $(n-s,L)$ list
decodable (i.e. every Hamming ball of radius $n-s$ has at most $L$
codewords in it) and one can do error detection on $C$ projected down to
any test of size at least $s$. 
If we pick our disjunct matrix carefully and $L$ is not too large, it
seems intuitive that one should be able to have, for most of the tests, that
$\vx_Q\neq\vzero$ implies $\vy_Q\not\in C_Q$. We are able to show that
if the matrix is picked at random, then this property holds. In addition, it is
the case that the ``decoding" of $\vx$ from the result of the test can be 
done even if some of the test results are faulty (i.e. $\vy_Q\in C_Q$ even
though $\vx_Q\neq\vzero$). The proof of this fact requires a fair bit of
work: we will come back to the issues in a bit.

Another side-effect of the fact that our algorithm does not readily translate 
into the group testing scenario is that even though we have been
able to salvage the case for $\wt(\vx)\le e$, we can no longer guarantee
that if $\wt(\vx)>e$, that our algorithm will catch it. In the latter case,
our algorithm might say $\wt(\vx)>e$ or it might return a subset
$S\subseteq [n]$ that purportedly contains all the error locations.
However, we can always check
if $\vy_{[n]\setminus S}\in \rs_{[n]\setminus S}[k]$ to rule out the latter
case.
This seems to require another pass on the input but we are able to
implement the final algorithm in one pass by giving a one-pass algorithm
for the following problem: Given as input $\vy$ followed by $T\subseteq [n]$
such that $|T|=e$, design a one-pass $O(e\log{n})$ space algorithm to check if
$\vy_{[n]\setminus T}\in \rs_{[n]\setminus T}[k]$. The main idea is to
encode the locations in $T$ as an unknown degree $e$ polynomial and  to
fill in the unknown coefficients once the algorithm gets to $T$ in the input.

We now return to the question of how many tests we can get away with
while using $e$-disjunct matrices. The best known construction uses
$O(e^2\log{n})$ tests~\cite{test-book} and this is tight to within
a $\log{e}$ factor (cf.~\cite{furedi}). Thus, to get sublinear space,
we need to have $e=o(\sqrt{n})$. To break the $\sqrt{n}$ barrier, 
instead of $e$-disjunct matrices, we use the recently discovered notion of the
$(e,e)$-list disjunct matrix~\cite{list-disj}. An $(e,e)$-list disjunct matrix
has the property that when applied to $\vx$ such that $\wt(\vx)\le e$,
it returns a subset $S\subseteq [n]$ such that (i) $x_i=1$ implies
$i\in S$ and (ii) $|S|\le \wt(\vx)+e$. It is known that such matrices
exist with $O(e\log{n})$ rows. In Section~\ref{app:list-disj}, we show
that such matrices can be constructed with $O(e\log^2{n})$ random bits.
However, note we can now only distinguish between the cases
of $\wt(\vx)\le e$ and $\wt(\vx)\ge 2e$.

The use of list disjunct matrices also complicates our result for worst
case errors that uses the list decodability of the code under consideration.
The issue is that when we pick the desired matrix at random, with the extra
task of ``avoiding" all of the $L-1$ codewords other than $\vc$
that can falsify the answer to the test, we can only guarantee that the
``decoding" procedure is able to recover a constant fraction of the
positions in error. This is similar to the notion of error reduction
in~\cite{spielman}. This suggests a natural, iterative $O(\log{e})$-pass
algorithm. Using our earlier trick, we can again implement our
algorithm in one pass. Finally, the plain vanilla proof needs
$\Omega(n)$ random bits. We observe that the proof goes through with
limited independence and use this to reduce the amount of randomness to
$O(e^2\log^3{n})$ bits. Reducing the random bits to something smaller,
such as $O(e\log{n})$, is an open problem.

The speedup in  the runtime from the naive $\tilde{O}(n^2)$
to $\tilde{O}(ne)$
for the tolerant testing algorithms is obtained by looking at certain explicit disjunct 
matrices and observing that the reduced error detection problems are
nicely structured.

There are two unsatisfactory aspects of our algorithms:
(i) The $O(e\log^2{n})$ space complexity and (ii) The condition
that $e\le O(n/k)$ (which in turn follows from the fact that we
have $s=n/(2e)$). We show that both of these shortcomings are essentially
unavoidable with our techniques. In particular, a lower bound on
the $1^+$ decision tree complexity of the threshold function
from~\cite{k-plus} implies that at least $\Omega(e)$ invocations of the
error detection routine are needed. Further, we show that for sublinear
test complexity, the support size $s$ must be in $O(\frac{n}e\log{n})$. This
follows by interpreting the reduction as a set cover problem and observing
that any set covers only a very small fraction of the universe.

\section{Preliminaries}
\label{sec:prelims}

We begin with some notation. Given an integer $m$, we will use $[m]$ to denote
the set $\{1,\dots,m\}$.
We will denote by $\F_q$ the finite field with $q$
elements. An $a\times b$ matrix $M$ over $\F_q$ will be called strongly
explicit if given any $(i,j)\in [a]\times [b]$, the entry $M_{i,j}$
can be computed in space $\mathrm{poly}(\log{q}+\log{a}+\log{b})$.
Given a vector $\vy\in\Sigma^n$ ($C\subseteq \Sigma^n$ resp.)
and a subset $S\subseteq [n]$, we will use $\vy_S$ ($C_S$ resp.) to denote
$\vy$ (vectors in $C$ resp.) projected down to the indices in $S$.
We will use $\wt(\vx)$ to denote the number of non-zero entries in
$\vx$. Further, for $S\subseteq [n]$, we will use $\wt_S(\vx)$ to denote
$\wt(\vx_S)$.

\vspace{2mm}
\noindent
\textbf{Codes.}
A code of {\em dimension} $k$ and {\em block length} $n$ over an
alphabet $\Sigma$ is a subset of $\Sigma^n$ of size $|\Sigma|^k$. The
{\em rate} of such a code equals $k/n$.  
 A code $C$ over
$\F_q$ is called a linear code if $C$ is a linear subspace of $\F_q^n$. 
If $C$ is linear, 
then it can be described by its parity-check
matrix $H$, i.e. for every $\vc\in C$, $H\cdot \vc^T=\vzero$.
An asymptotically good code has constant rate and constant relative distance
(i.e. any two codewords differ in at least some fixed constant fraction
of positions).

\vspace*{2mm}
\noindent
\textbf{Tolerant Testers.}
We begin with the central definition. Given a code $C\subseteq\Sigma^n$,
 reals $0\le d<c\le 1$, $0\le \eps_1<\eps_2\le 1$ and integers $r=r(n)$ and
$s=s(n)$, an $(r,s,\eps_1,\eps_2)_{c,d}$-\ttds\ $\mathcal{T}$
for $C$ is a randomized
algorithm with the following properties for any input $\vy\in\Sigma^n$:
{ (1) If $\Delta(\vy,C)\le \eps_1 n$, then $\mathcal{T}$ accepts with probability at least $c$;}
{(2) If $\Delta(\vy,C)\ge \eps_2 n$, then $\mathcal{T}$ accepts with probability at most $d$; }
{(3) $\mathcal{T}$ makes at most $r$ passes over $\vy$; and }
{(4) $\mathcal{T}$ uses at most $s$ space for its computation. }


Further, we will consider the following special cases of
an $(r,s,\eps_1,\eps_2)_{c,d}$-\ttds:
{(i) An $(r,s,0,\eps)_{c,d}$-\ttds\ will be called an $(r,s,\eps)_{c,d}$-\ltds.} 
{(ii) An $(r,s,0,1/n)_{c,d}$-\ttds\ will be called an $(r,s)_{c,d}$-\edds.}
%
There are some  definitional issues that are resolved in
Section~\ref{app:def}.

\vspace*{2mm}
\noindent
\textbf{List Disjunct Matrices.}
We give a low-space algorithm that can compute a small set of possible defectives
given an outcome vector which is generated by a list disjunct matrix.  
Relevant
definitions and material related to the algorithm can be found in Section~\ref{sec:twop}.

\vspace*{2mm}
\noindent
\textbf{Some Explicit Families of Codes.}
We now mention two explicit families of codes that we will
see later on in the paper.
We first begin with the Reed-Solomon code. Given $q\ge n\ge 1$ and a subset
$S=\{\alpha_1,\dots,\alpha_n\}\subseteq \F_q$, the Reed-Solomon code with 
{\em evaluation set}
$S$ and  dimension $k$, denoted by $\rs_S[k]$, is defined as follows: Any message
in $\F_q^k$ naturally defines a polynomial $P(X)$ of degree at most $k-1$ over
$\F_q$. The codeword corresponding to the message is obtained by evaluating
$P(X)$ over all elements in $S$. It is known that
a $(n-k)\times n$ parity check matrix of $\rs_S$ is given by 
$H_{\rs_S}=\{v_j\cdot\alpha_j^i\}_{i=0~~~~,~j=1}^{n-k-1,~n},$
where

$$v_j=\frac{1}{\prod_{1\le\ell\le n,\ell\neq j}\left(\alpha_j-\alpha_{\ell}\right)}$$.

Another explicit code family we will consider are expander codes. These are
binary codes whose parity check matrices are incidence matrices of 
constant-degree bipartite expanders. In particular, if we start with a
strongly explicit expander, then the parity check matrix of the corresponding
expander code will also be strongly explicit. 

\section{Connections to Property Testing}
\label{app:prop-test}

\subsection{The basic connection}
We now highlight a simple connection between tolerant testers in the data
stream world and tolerant testers in the query world:

\begin{prop}
\label{prop:basic-conn}
Let $C\subseteq [q]^n$ be such that it has a tolerant tester
$\mathcal{T}$ with query complexity $r$, thresholds $\eps_1$ and $\eps_2$
and time complexity $t_q(r)$ (i.e. it makes $t_q(r)$ operations
over $[q]$ for any possible query realization).
Then there also exists an $(r,O(t_q(r)+r\log{n}),\eps_1,\eps_2)_{c,s}$-\ttds\ 
$\mathcal{T}'$. Further, if $\mathcal{T}$ is {\em non-adaptive}, then
$\mathcal{T'}$ can be implemented as a $(1,O(t_q(r)+r\log{n}),\eps_1,\eps_2)_{c,s}$-\ttds.
\end{prop}
\begin{proof}
The claimed result follows from the obvious simulation. In general, the tester
$\mathcal{T}'$ works as follows:
 $\mathcal{T}$ queries $r$ positions in the input and
then applies some function on the queried values (in the case when
$\mathcal{T}'$ is adaptive, it applies possibly $r$ different functions
after each query). As the total time complexity of $\mathcal{T}$ is $t_q(r)$,
the entire computation of $\mathcal{T}$ can be done in time (and hence, space)
$t_q(r)$.  If $\mathcal{T}$ is non-adaptive, all the query positions can be 
decided upfront and all the values can be determined in one pass. Otherwise
the simulation might need $r$ passes in the worst case. We might need to use
an additional $O(r\log{n})$ space to store the indices of the query
positions, which implies that $\mathcal{T}'$ has the claimed properties.
\end{proof}

\begin{remark}
In general, one cannot say much about $t_q(r)$ other than bounding it by 
$2^{O(q^r)}$ as the definition of the usual tolerant tester does not put any computational
efficient constraints on the testers. However,
if the \ttq\ $\mathcal{T}$ makes a constant number of queries then its time complexity is also a constant number of operations over the alphabet.
\end{remark}

\subsection{Tolerant Testing of Reed-Muller Codes}
\label{app:rm}

Let $\grm(q,\ell,m)$ denote the Reed-Muller code obtained by evaluating
$m$-variate polynomials over $\F_q$ of total degree $\ell<q$. 
These codes are known to have block length $n=q^m$, dimension $\binom{m+\ell}{m}$ and
distance $(1-\ell/q)n$ (cf.~\cite[Lect. 4]{sudan-notes}). Note that this implies
that if $m$ is a constant and $\ell=\Omega(q)$, then $\grm(q,\ell,m)$ is
asymptotically good.
These codes
are known to be tolerant testable in the property testing world with
polynomial number of queries.

\begin{theorem}[\cite{GR-tolerant}]
Let $m,\ell,q\ge 1$ be integer such that $\ell<c\cdot q$ for some
universal constant $c$. Then there exists a
tolerant tester for $\grm(q,\ell,m)$ in the property testing world that can distinguish between
at most $\eps_1 n$ and at least $\eps_2 n$ errors with $q=n^{1/m}$
queries, where $\eps_1$ and
$\eps_2$ are absolute constants that only depends on $c$.
\end{theorem}

In fact, the test is simple to describe: pick a random line in $\F_{q^m}$
and check if the projected down received word is within some threshold
Hamming distance from the corresponding RS code of dimension $\ell$ and
block length $q$. This latter step can be solved using the fast
list decoding algorithm for RS codes from~\cite{rs-fast} in time
$O(q\log^2{q})$ (and hence, in the same amount of space). Thus, Proposition~\ref{prop:basic-conn}
implies the existence of an asymptotically good code that can
be tolerant tested by a one-pass, $O(n^{\eps})$ space (for any $\eps>0$)
data stream algorithm. (Note that the code has an absolute constant as its
relative distance but its rate is exponentially small in $1/\eps$.)

\section{Some definitional issues}
\label{app:def}
One decision that we need to make is how we count the space/time requirement
for our algorithms. In particular, given a code defined over $\Sigma$, do
we do our accounting in terms of number of operations over $\Sigma$ or the
number of operations over ``bits"? This question is moot when $\Sigma$ has
constant size as both the measures will be within constant factors
of each other. However, if $|\Sigma|$ can depend on $n$, which will be the
case in some of the codes that we consider in this paper, the two measures
will not be within constant factors anymore. In particular, for
arbitrary $\Sigma$, an operation over $\Sigma$
may take $\Omega(n)$ space and time,
which will be prohibitive for our purposes.

We resolve the question above in the following way: First, we will account
for the complexity measures in terms of the number of operations in $\Sigma$.
Further, for positive results, we will focus on the case where 
$\Sigma=\F_q$ with $q\le n^{O(1)}$. Note that in this case, all
operations (including addition, multiplication and exponentiation) can be
done in $\mathrm{poly}(\log{n})$ time and $O(\log{n})$ space. Finally, for
general fields $\F_q$, the algorithm will also need access to an
irreducible polynomial (of degree at most $O(\log{n})$). However, note
that the definition of a code provides the definition of its 
alphabet. 
We will assume that the algorithm has full prior knowledge about the
code (including e.g. the value of $n$).
Thus, we will assume that the irreducible polynomial
will be (implicitly) a part of the input to the algorithm. For certain cases,
when the irreducible polynomial is part of an explicit family,
the algorithms can compute these
irreducible polynomials ``on the fly" and thus, do not need to be
part of the input.

Finally, by definition, the block length of a code is fixed. 
However, for
a meaningful asymptotic analysis, we need to think of an increasing
sequence of block lengths. Thus, from now on when we talk about a code,
we implicitly mean a family of codes.

\section{Data Stream Algorithms for Error-Detection}
\label{sec:ed}

\vspace*{2mm}
\noindent
\textbf{A positive result.}
We first show that any linear code with a strongly explicit
parity check matrix has an efficient $1$-pass data stream
error detector. 

Note that for a linear code $C\subseteq\F_q^n$ with parity check matrix $H$,
the error detection problem with the usual polynomial time complexity
setting is trivial. This is because by the definition of parity check
matrix for any $\vy\in\F_q^n$, $\vy\in C$ if and only if $H\cdot \vy^T=\vzero$.
However, the
naive implementation requires $\Omega(n)$ space which is prohibitive
for data stream algorithms. We will show later  
that for deterministic data stream algorithms with a constant number of
passes, this space requirement is unavoidable for asymptotically good codes.

However, the story is completely different for randomized algorithms. If we are
given the
{\em syndrome} $\vs=H\vy^T$ 
instead of $\vy$ as the input, then we just have to solve
the {\em set equality} problem which has a very well-known one-pass
$O(\log{n})$-space data stream algorithm based on the 
{\em finger-printing} method.
Because $\vs$ is a fixed linear combination of
$\vy$ (as $H$ is known), we can use the
fingerprinting technique in our case.
Further, unlike the usual fingerprinting method, which requires
any large enough field, we need to use an extension field
of $\F_q$.
For this, we need to get our hands on irreducible polynomials over $\F_q$.

\subsection{Families of Irreducible Polynomials}

In our error detection algorithm we need low space construction of families of irreducible polynomials. 
Since our final algorithm will be randomized, a randomized algorithm to construct irreducible polynomials works.
The following result is well-known (cf.~\cite[Chap. 20]{shoup-book}:
\begin{theorem}
\label{thm:irr-rand}
Let $q$ be a prime power, $d$ be an integer and $0<\delta<1$ be a real number.
Then there exists a randomized algorithm that outputs an irreducible
polynomial of degree $d$ over $\F_q$ with probability at least $1-\delta$.
Further, this algorithm makes $O(d^4\log(1/\delta)\log{q})$ operations over
$\F_q$ and needs $O(\log(1/\delta)+d\log{q})$ bits of space.
\end{theorem}

Coming up with a deterministic polynomial algorithm for construction
of irreducible polynomials is an open question.
However, it turns out that in our application, we would be happy if
the final irreducible polynomial has degree $d'$ such that $d'\ge d$
and is not much larger than $d$. In particular for
\textit{prime} $p$, there exists a deterministic algorithm
that runs in time (and hence, space) $\poly{d\log{p}}$
and outputs a polynomial with degree at least $d$ and at most $O(d\log{p})$
\cite{adleman-lenstra}.

Next, we show that constructing such irreducible polynomials
can also be done for fields of characteristic $2$.
The result
follows from
other known results.

\begin{theorem}
\label{thm:irr-exp}
Let $q$ be a power of $2$ and let $d\ge 1$ be an integer. 
Given the
irreducible polynomial that generates $\F_q$, 
there exists a
deterministic
$O(d\log{q})$ space, $O((d^2+\log{q})\log^2{q})$-time algorithm that computes an irreducible polynomial
over $\F_q$ with degree $d'$ such that $d\le d'\le 2d$.
\end{theorem}

Coming up with an analogous result to Theorem~\ref{thm:irr-exp} for odd characteristic
seems to be an open problem.

We begin the proof of Theorem~\ref{thm:irr-exp}.
We will use the following result:

\begin{theorem}[cf.~\cite{gao-book}] 
\label{thm:irr-exp-even}
Let $m\ge 1$ be an integer and let
$\beta\in\F_{2^m}$ such that $Tr(\beta)\neq 0$, where $Tr(x)=\sum_{i=0}^{m-1} x^{2^i}$ is the trace function. Define the polynomials $A_k(X)$ and
$B_k(X)$ recursively as follows (for $k\ge 0$):
\begin{align*}
A_0(X)&=~X\\
B_0(X)&=~1\\
A_{k+1}(X)&=~ A_k(X)B_k(X)\\
B_{k+1}&=~ A_k^2(X)+B_k^2(X).
\end{align*}
Then $A_k(X)+\beta\cdot B_k(X)$ is an irreducible polynomial
over $\F_{2^m}$ of degree $2^k$.
\end{theorem}

To begin with, let us assume we can get our hands on a $\beta$ as 
required in Theorem~\ref{thm:irr-exp-even}. Given such a $\beta$, the rest
of the proof is simple. Pick $k$ to be the smallest integer such that $d'=2^k\ge d$. It is easy to check that $d\le d'\le 2d$ as required. To compute
the final irreducible polynomial, we will need to do $k$ iterations to
compute $A_i(X)$ and $B_i(X)$ (for $1\le i\le k$). It is easy to check that
each iteration requires $O(2^i\log{q})$ space (to store the intermediate
polynomials) and $O(2^{2i}\log^2{q})$ time (to compute the product of two polynomials
of degree at most $2^i$). To complete the proof of Theorem~\ref{thm:irr-exp},
we show how to efficiently compute
an appropriate $\beta$. 

We claim that $\beta$ can be
chosen to be $\alpha^i$ for some $0\le i\le m-1$, where
we use $\{1,\alpha,\dots,\alpha^{m-1}\}$ as the standard basis for 
$\F_{2^m}$, for some root $\alpha$ of the irreducible polynomial that
generates $\F_{2^m}$.
\footnote{
Note that if $m$ is odd, then just $\beta=1$ suffices.} To see why this is
true, assume for the sake of contradiction that $Tr(\alpha^i)=0$ for every
$0\le i\le m-1$. Then as every $\gamma\in\F_{2^m}$ can be written as a
linear combination of $1,\alpha,\dots,\alpha^{m-1}$, $Tr(\gamma)=0$ (this follows
from the well-known fact that $Tr(\gamma_1+\gamma_2)=Tr(\gamma_1)+Tr(\gamma_2)$).
This implies that $Tr(X)$ has $2^m$ roots, which is a contradiction as
$Tr(X)$ is a non-zero polynomial of degree $2^{m-1}$. Finally, the correct
choice of $\beta=\alpha^i$ can be determined by going through
all $0\le i\le m-1$ and evaluating $Tr(\alpha^i)$ (which can be done in 
$O(\log^3{q})$ time).


We now state our result.

\begin{theorem}
\label{thm:alg-ed}
Let $C\subseteq\F_q^n$ be a linear code of dimension $k$
and block length $n$
with parity check
matrix $H=\{h_{i,j}\}_{i=0~~~~,j=1}^{n-k-1,n}$. Further, assume that
any entry
$h_{i,j}$ can be computed in space $\mathcal{S}(n,q)$, for some
function $\mathcal{S}$.  
Given an $a\ge 1$,
there exists a 
$(1,O(\mathcal{S}(n,q)+a\log{n}))_{1,n^{-a}}$-\edds\ for $C$. 
\end{theorem}

\subsection{Proof of Theorem~\ref{thm:alg-ed}}
Let $\vy=(y_1,\dots,y_n)\in\F_q^n$ be the received word and let 
$\vs=(s_0,\dots,s_{n-k-1})=H\vy^T\in\F_q^{n-k}$. It is easy to check
that for every $0\le i<n-k$,
\begin{equation}
\label{eq:s-vals}
s_i=\sum_{j=1}^n y_jh_{i,j}.
\end{equation}
Further define
\[S(X)=\sum_{i=0}^{n-k-1} s_iX^i.\]
Note that our task is to verify whether $S(X)$ is the all zeros polynomial.
Towards this end, we will use the fingerprinting technique.

Let $Q=q^d$ for some $d$ to be chosen later such that $n^{1+\a}\le Q\le qn^{1+\a}$. The
algorithm is simple: pick a random $\beta\in\F_Q$ and verify if
$S(\beta)=\sum_{i=0} s_i\beta^i=0$. By (\ref{eq:s-vals}), 
this is the same as checking if $\sum_{i=0}\left(\sum_{j=1}^n y_jh_{i,j}\right)\beta^i=0$. Note that this is possible as $\F_Q$ is an extension field of $\F_q$
and thus, all the terms
in the sum belong to $\F_Q$. Thus, by changing the order of sums, 
we need to verify if
\begin{equation}
\label{eq:ed-check}
\sum_{j=1}^n y_j\left(\sum_{i=0}^{n-k-1}\beta^ih_{i,j}\right)=0.
\end{equation}

It is easy to verify that the above sum can be computed in one pass as long
as the quantity $\sum_{i=0}^{n-k-1}\beta^ih_{i,j}$ can be computed efficiently
``on the fly." The latter is possible as we know $\beta$ and we can compute
any entry $h_{i,j}$ on the fly.
If $S(X)$ is the all zeros polynomial, then the check will
always pass. If on the other hand, $S(X)$ is a non-zero polynomial of degree
at most $n$, then $S(\beta)=0$ for at most $n$ values $\beta\in\F_Q$. Thus,
the probability of the check passing is at most $n/Q$ which by our choice of
$Q$ is at most $1/n^{\a}$.

To complete the proof we need to analyze the space requirement of the
algorithm above. First we note that by
Theorem~\ref{thm:irr-rand} we can compute an irreducible polynomial
over $\F_q$ of degree $d
=\left\lceil 2 \frac{\log{n}}{\log{q}}\right\rceil$.
Note that $n^{1+\a}\le Q\le qn^{1+\a}$ as claimed before. Also note that any operation in
$\F_Q$ can be carried out by storing $O(d)$ elements from $\F_q$. This implies that
the sum in (\ref{eq:ed-check}) can be computed with space $O(\mathcal{S}(n,q)+d)=O(\mathcal{S}(n,q)+\alpha\log{n})$.

An inspection of the proof above shows that the time complexity of
the \edds\ is dominated by the number of $\F_q$ operations needed to
compute $\sum_{i=0}^{n-k-1}\beta^j h_{i,j}$.

For expander codes, this time complexity is just a constant number of
$\F_Q$ operations (and hence $O(\log{n}/\log{q})$ operations in $\F_q$).
For RS codes, recall that we have $h_{i,j}=v_j\cdot\alpha_j^i,$ where
$v_j=\frac{1}{\prod_{1\le \ell\le n, \ell\neq j}(\alpha_j-\alpha)}$. 
If, say, for some fixed $\beta\in\F_q^*$, $v_j=\beta$ for every $1\le j\le n$,
then one can essentially ignore $v_j$ and one only needs to  
compute
$\sum_{i=0}^{n-k-1}\beta^j \alpha_j^i$, which
is
just $\frac{(\beta\alpha_j)^{n-k}-1}{\beta\alpha_j-1}$ unless
$\beta=0$ (in which case the sum is $0$) or $\beta=(\alpha_j)^{-1}$
(in which case the sum is just $(n-k)$ modulo the characteristic of
$\F_q$). The latter condition can be verified with $\poly{\log{n}/\log{q}}$
operations in $\F_q$.

In general RS, any $h_{i,j}$ can be computed with $\tilde{O}(n)$ operations in
$\F_q$. Thus, the sum can be computed in time $\tilde{O}(n^2)$.

Thus, we have argued that
\begin{corollary}
\label{cor:detect-time}
Let $q$ be a prime power and define $S=\{\alpha_1,\dots,\alpha_n\}$. Then there
exists an $(1,O(\log{n})_{1,1/2}$-\edds\ for $\rs_S$ that runs in time 
$\tilde{O}(n^2)$. Further, if there exists a $\beta\in\F_q^*$ such that for every
$1\le j\le n$, $\prod_{1\le \ell\le n, \ell\neq j}(\alpha_j-\alpha)=\beta$,
then the algorithm can be implemented in $\tilde{O}(n)$ time.
\end{corollary}

It is easy to check that $\mathcal{S}(q,n)$ is $O(\log{n})$ for
(strongly explicit) expander codes and RS (and GRS) codes. This
implies the following:
\begin{corollary}
Let $n\ge 1$. Then for $q=2$ and $n\le q\le \poly{n}$, there exists
an asymptotically good code $C\subseteq\F_q^n$ that has
a $(1,O(\log{n}))_{1,1/2}$-\edds.
\end{corollary} 

\vspace*{2mm}
\noindent
\textbf{A negative result.}
We show that randomness is necessary even for local testing. In 
particular, we show the following:

\begin{theorem}
\label{thm:det-lt}
Let $C\subseteq [q]^{n}$ be a code of rate $R$ and relative distance $\delta$
and let $0\le\eps\le \delta^2/8$ be a real number. Then 
any $(r,s,\eps)_{1,0}$-\ltds\ for $C$ needs to satisfy
$r\cdot s\ge \frac{\delta Rn}{6}$.
\end{theorem}

\subsection{Proof of Theorem~\ref{thm:det-lt}}

We will be using communication complexity to prove Theorem~\ref{thm:det-lt}.

The proof uses he standard fooling set technique, however, unlike the
usual application of two-party communication complexity in
data stream algorithms, where the stream is broken up into two fixed
portions, in our case we need to be careful about how we divide up the input.
To see the necessity of this, consider the code
$C\times C\subseteq \Sigma^{2n}$ and say we break the received word
$\vy$ in the middle and assign the first half (call it $\vy_1$) to
Alice and the second half to Bob. In this case there is a simple
$O(\log{n})$ protocol-- Alice simply sends the distance of $\vy_1$ to
the closest codeword in $C$ to Bob-- to compute the distance of $\vy$
to the closest codeword in $C\times C$ \textit{exactly}. However, we
can show that for every asymptotically good code, there is some way of
breaking up the input into two parts such that there exists an
exponentially sized fooling set.

To further explain, we do a quick recap of some of the basic
concepts in communication complexity and refer the reader to source
material for more details~\cite{cc-book}.

Let $g:\{0,1\}^{n_1}\times\{0,1\}^{n_2}\rightarrow \{0,1\}$ be a function.
Further assume Alice has $x\in\{0,1\}^{n_1}$ and Bob has $y\in\{0,1\}^{n_2}$.
The (deterministic) communication complexity of $g$, denoted by $\cc(g)$, 
is the minimum
number of bits that Alice and Bob must exchange in order to determine
$g(x,y)$ in the worst case. The following observation is a standard technique
to obtain lower bounds for data stream algorithms:

\begin{prop}
Let $\mathcal{A}$ be an $r$-pass, $s$-space deterministic
data stream algorithm that decides $g$. Then $r\cdot s\ge \cc(g)$.
\end{prop}

Next we consider the following technique for lower bounding the communication
complexity of a function. A subset $F\subseteq \{0,1\}^{n_1}\times \{0,1\}^{n_2}$ is
called a {\em fooling set} for $g$ if (i) For every $(x,y)\in F$, $g(x,y)=b$
for some fixed $b\in \{0,1\}$ and (ii) For every $(x_1,y_1)\neq (x_2,y_2)\in F$,
either $g(x_1,y_2)=1-b$ or $g(x_2,y_1)=1-b$. The following result is
well-known:
\begin{prop}[cf. \cite{cc-book}]
\label{prop:fool}
Let $g:\{0,1\}^{n_1}\times\{0,1\}^{n_2}\rightarrow \{0,1\}$ and $F$ be a 
fooling set for $g$. Then $\cc(g)\ge \log\left(|F|\right)$.
\end{prop}

Finally, we will consider boolean functions with one input
and we define their communication complexity as follows: Let
$f:\{0,1\}^n\rightarrow\{0,1\}$. Further, for any $0\le n_1,n_2\le n$ such that
$n_1+n_2=n$, define $f_{n_1,n_2}:\{0,1\}^{n_1}\times \{0,1\}^{n_2}\rightarrow
\{0,1\}$ by naturally ``dividing" up the $n$-bit input for $f$ into the two required
inputs for $f_{n_1,n_2}$. 
The communication complexity of $f$ is then
defined as follows:
\[\cc(f)=\max_{\substack{0\le n_1,n_2\le n,\\ n_1+n_2=n}} \cc\left(f_{n_1,n_2}\right).\]

We are now ready to prove Theorem~\ref{thm:det-lt}. We will
do so by proving that the deterministic communication complexity of the 
following function is large. Define $f_C:\{0,1\}^n\rightarrow
\{0,1\}$ such that $f_C(\vy)=1$ if $\vy\in C$; $f_C(\vy)=0$ if $\Delta(\vy,C)\ge \delta^2n/8$;
otherwise $f_C(\vy)$ can take an arbitrary value. We will show that:
\begin{lemma}
\label{lem:large-fool}
$f_C$ has a fooling set of size at least $q^{\delta Rn/6}$.
\end{lemma}
Note that as we are measuring space in terms of the number of elements from
$[q]$, Lemma~\ref{lem:large-fool} and Proposition~\ref{prop:fool} imply Theorem
\ref{thm:det-lt}.

In the rest of the section, we prove Lemma~\ref{lem:large-fool}. For notational
convenience, define $k=Rn$, $d=\delta n$, $\a=\frac{\delta}{4}$ and
$\beta=\frac{\delta}{6}$. Thus we need to exhibit a fooling set of size at least
$q^{\beta k}$.

Our fooling set will be a subset $F\subseteq C$ with an
 $0<n_1<n$ such that $F_{[n_1]}$  and $F_{[n]\setminus [n_1]}$
have distance at least $\a d/2$. Further for any $\vc\in T$ and $\vc'\in C\setminus F$,
$\Delta(\vc,\vc')\ge \a d/2$. We claim that such an $F$ is indeed a fooling set.
To see this consider $\vc^1\neq\vc^2\in F$, where $\vc^1=(c^1_1,c^1_2)$,
$\vc^2=(c^2_1,c^2_2)$, $c^i_1\in F_{[n_1]}$ and $c^i_2\in F_{[n]\setminus [n_1]}$.
By definition, $f_C(\vc^1)=f_C(\vc^2)=1$. Next we show that either $f_C(\vy_1)=0$
or $f_C(\vy_2)=0$,
where $\vy_1=(c^1_1,c^2_2)$ and $\vy_2=(c^2_1,c^1_2)$.
For any $\vc\in C\setminus F$, this is true by definition of $F$. For
any $\vc\in F$, this is true by the distance properties of $F_{[n_1]}$
and $F_{[n]\setminus [n_1]}$.

We will construct the fooling set $F$ in a greedy fashion. We begin with the
case when $n<2(1-\a)d$. We claim that in this case 
$F=C$ and $n_1=\lfloor n/2\rfloor$
works. Note that both $n_1$ and $n-n_1$ are both at most $(1-\a)d$. Since
$C$ has distance $d$, this implies that both $F_{[n_1]}$ and $F_{[n]\setminus
[n_1]}$ have distance at least $d-(1-\a)d=\a d$. This completes the
proof for the base case.

For the general $n\ge 2(1-\a)d$ case, we reduce it to the base case. In
particular, we present a greedy iterative process, where at the end of
the $i^{th}$ step, we have a subset $F_i\subseteq C$, with the property
that for every $\vc\in F_i$ and $\vc'\in C\setminus F_i$,
$\Delta(\vc,\vc')\ge \a d/2$. We of course start with $F_0=C$.
It will turn out that we will run this
process for $r\le \frac{1}{2\a}$ times and $F=F_r$. For ease of exposition,
we will also track variables $m_i$ and $d_i$ such that $m_0=n$
and $d_0=d$. 
Next, we mention the invariance that we will maintain with the iterative
algorithm. First, it will always be the case that $m_{i+1}=m_i-(1-\a)d_i$
and $d_{i+1}\ge d_i-\a d_i$. Think of $m_i$ as the block length of the
(projected down) $F_i$ and $d_i$ as the corresponding distance.

Assume that we have our hands on $F_i$. 
If $m_i<2(1-\a)d_i$, then we are in the base case and the process terminates.
(In this case $r=i$ and $n_1=\lfloor m_r/2\rfloor$.)
If not, then $F_i$
projected onto the first $m_i$ positions (call this projected
down code $G_i$) has distance at least $d_i$. Group codewords
in $F_i$ such that in each cluster the codewords in $G_i$ projected down
to the last $(1-\a)d_i$ positions differ from each other in $<\a d_i$
positions. If the number of clusters is at least $q^{\beta k}$,
then let $F_{i+1}$ be defined by picking one codeword from each of the
clusters and the process terminates. (In this case $r=i+1$
and $n_1=m_r$.) If not, then define $F_{i+1}$
to be the largest cluster and define $m_{i+1}=m_i-(1-\a)d_i$.
We then continue the process for $i+1$.

For the time being, assume that the following are true: (i) The
process stops at iteration $r$ such that $r\le \frac{1}{2\a}$;
(ii) $G_i$ has distance $d_i$ such that $d_i\ge (1-i\a)d$; and
(iii) Every codeword in $F_i$ differs from every codeword in
$F_{i-1}\setminus F_i$ in at least $\a d/2$ positions.
Assuming these three properties, we argue that $F_r$
indeed has the required properties.

Assume that the process terminates when the base case is reached.
Let $G=G_r$. Then by the argument in the base case, we have
that both $G_{[n_1]}$ and $G_{[m_r]\setminus [n_1]}$
(and hence, $F_{[n_1]}$ and $F_{[n]\setminus [n_1]}$) have distance
at least $\a d_r\ge \a d/2$, where the inequality follows
from properties (i) and (ii). Further, by property (iii), it is the case that
for every $\vc\in F$ and $\vc'\in C\setminus F$, $\Delta(\vc,\vc')\ge \a d/2$.
Finally, note that when we pick a single cluster, we have $|F_{i+1}|\ge |F_i|/q^{\b k}$.
Thus, if we terminate with the base case, we have
\[|F|\ge \frac{q^k}{q^{\beta k}}=q^{k(1-r\beta)}\ge q^{\beta k},\]
where the last inequality follows from the following argument for the inequality
$1-r\beta\ge \beta$. This inequality is satisfied if
\[r\le \frac{1}{\beta}-1.\]
Now by property (i), $r\le \frac{1}{2\a}=\frac{2}{\delta}=\frac{3}{\beta}$.
Now as $\beta\le \delta/6 \le 1/6$, we have $\beta\le 1/(3\beta)\le 1/\beta-1$,
as desired.

Now we consider the case when the process terminates before reaching the 
base case. In this case because of the
termination condition, we have that $G_{r-1}$ projected onto the last
$(1-\a)d_{r-1}$ positions (and hence, $F_{[n]\setminus [n_1]}$) has distance
at least $\a d_{r-1}\ge \a d/2$, where the inequality follows from
properties (i) and (ii). Also $F_{[n_1]}$ has distance at least $d_{r-1}-(1-\a)d_{r-1}
=\a d_{r-1}\ge \a d/2$. Further, by property (iii), it is the case that
for every $\vc\in F$ and $\vc'\in C\setminus F$, $\Delta(\vc,\vc')\ge \a d/2$.
Finally, by the termination condition, we have $|F|\ge q^{\beta k}$, as
desired.

Thus, we are done with the proof modulo showing that properties (i)-(iii) hold,
which is what we do next. We begin with property (i). Note that if we do not terminate
in the middle, then we have $d_{i+1}\ge (1-\a)d_i\ge (1-\a)^i d$. Now note that
\[m_r=n-\sum_{i=0}^{r-1} (1-\a)d_i\le n-d\sum_{i=1}^r (1-\a)=n-\frac{(1-\a)(1-(1-\a)^r)d}{\a}.\]
Since $m_r\ge 0$, we have
\[1-(1-\a)^r\le \frac{\a n}{(1-\a) d}=\frac{\a}{(1-\a)\delta}\le 1-\exp(-1/2),\]
where the last inequality follows from the fact that $\a=\delta/4\le 1/4$.
Thus the above implies that
\[(1-\a)^r\ge \exp(-1/2),\]
which in turn implies
\[r\ln\left(\frac{1}{1-\a}\right)\le\frac{1}{2}.\]
Using the fact that $\ln(1-x)=-(x+x^2/2+x^3/3+\dots)$ for $|x|<1$, we get that the
above implies
\[r\left(\a+\a^2/2+\a^3/3+\dots\right)\le \frac{1}{2},\]
which in turn implies that
$r\a\le \frac{1}{2}$,
as desired.

We now move to property (ii). As we saw earlier, we have $d_i\ge (1-\a)^i d$ for $i\le r$. Now
as $\a r\le 1/2$ (and hence $\a i\le 1/2$), we have that $(1-\a)^i\ge 1-i\a$, which
proves property (ii). Note that this also implies that $d_i\ge d/2$. Finally,
for property (iii), note that by construction, if we do not terminate in middle at step $i$,
every codeword in $F_{i+1}$ differs
from $F_{i}\setminus F_{i+1}$ in at least $\a d_i$ positions. Since, $d_i\ge d/2$, property
(iii) follows. The proof is complete.

\vspace*{2mm}
\noindent
\textbf{Error detection of a projected down code.}
We will be dealing with $\rs_S[k]$
with
$S=\{\a_1,\dots,\a_n\}$. 
In particular, we are interested in a one-pass, low space data stream 
algorithm to solve the following problem: The input is
$\vy\in\F_q^n$ followed by a  subset $E\subseteq S$ with $|E|=e$. We
need to figure out if $\vy_{S\setminus E}\in\rs_{S\setminus E}[k]$.
We have the following result:

\begin{lemma}
\label{lem:weird-ed}
Let $e,n,k\ge 1$ be integers such that $k+e\le n$. Then the problem above
can be solved by a one-pass, $O(e+a \log{n})$ space data stream algorithm
with probability at least $1-n^{-a}$, for any $a \ge 1$.
\end{lemma}

\subsection{Proof of Lemma~\ref{lem:weird-ed}}

Consider the degree $e$ polynomial 
$P_E(X)=\prod_{i\in E} (X-\a_i)$. Further, consider a new received word
$\vz=(z_1,\dots,z_n)$ where $z_i=y_i\cdot P_E(\a_i)$. 
The algorithm to solve the problem above just checks to see if
$\vz\in\rs_S[e+k]$. 

We begin with the correctness of the algorithm above.
If $\vy_{[n]\setminus E}\in RS_{S\setminus E}[k]$, that is, $\vy_{[n]\setminus E}$
is the evaluation of a polynomial $f(X)$ of degree at most $k-1$ over
$S\setminus E$, then $\vz$ is the evaluation of
$f(X)\cdot P_E(X)$ over $S$. In other words, $\vz\in RS_{S}[k+e]$.

Now it turns out that the other direction is also true. That is, if
$\vz\in RS_{S}[k+e]$ then $\vy_{[n]\setminus E}\in\rs_{[n]\setminus E}[k]$. Note that
$\vz$ is the evaluation of a degree at most
$e+k-1$ polynomial $g(X)$ over $S$,
where $g(X)=P_E(X)\cdot h(X)$, where $h(X)$ has degree at most $k-1$. This is
easy to see: by definition $P_E(X)|g(X)$ and the degree requirement on 
$g(X)$ implies that $h(X)$ has degree at most $k-1$. Finally, as
$P_E(\alpha_i)\neq 0$ for $i\not\in E$, this implies that $h(\alpha_i)=y_i$
for $i\not\in E$. In other words, $\vy_{[n]\setminus e}\in \rs_{[n]\setminus E}$.

We conclude this proof by showing how to deal with the unknown $E$ using a
data stream algorithm. Since $E$ is unknown, let us denote $P_E(X)=X^e+\sum_{i=0}^{e-1} p_iX^i$, where $\{p_i\}$ are the unknown coefficients.
Recall that in our error detection algorithm to check if 
$\vz\in RS_S[e+k]$ we need to check if the following sum is $0$:
\[\sum_{j=1}^n y_jP_E(\a_j)\left(\sum_{i=0}^{n-k-e} \beta^ih_{i,j}\right),\]
where $\beta$ is a random element in a large enough extension field of
$\F_q$ and $\{h_{i,j}\}$ is the parity check matrix of $RS_S[k]$. Note that
as $P_E(X)=X^e+\sum_{i-0}^{e-1} p_iX^i$, the sum above can be written as $Q_e+\sum_{b=0}^{e-1} p_bQ_b$, where (for $0\le b\le e$)
\[Q_b=\sum_{j=1}^n y_j\a_j^b\left(\sum_{i=0}^{n-k-e} \beta^ih_{i,j}\right).\]
Note that each of the $Q_b$ sums can be computed in one pass and low space
without the knowledge of $E$. 

Now the algorithm to check if $\vz\in RS_S[k+e]$ is clear: maintain the
$e+1$ sums $\hat{Q}_b$. At the end of the pass, the previous algorithm 
knows the set $E\subseteq [n]$.
Given this, we can compute the coefficients $\{\hat{e}_b\}_{b=0}^{e-1}$. Then
we declare $\vy_{[n]\setminus E}\in \rs_{[n]\setminus E}$ if and only if
$\hat{Q}_t+\sum_{b=0}^{t-1} \hat{e}_b\hat{Q}_b=0$.
This extra computation will need storage for $O(e)$ elements in the extension
field of $\F_q$.

\section{Tolerant testing}
\label{sec:twop}

In this section we assume that we are working with $\rs_S[k]$,
where $S=\{\a_1,\dots,\a_n\}\subseteq \F_q$. (However,
our results will also hold for closely related codes such as
the folded RS code~\cite{GR-capacity}.)


As was mentioned in the Introduction, there is a trivial reduction from one
tolerant testing instance (say where we are interested in at most $e$ vs. $>e$
errors) to $\binom{n}{e}$ instances of error detection:
for each of the $\binom{n}{e}$ potential error locations, project the received word outside of
those indices and check to see if
it's a codeword in the corresponding RS code
via the  algorithm in Theorem~\ref{thm:alg-ed}. Using
Theorem~\ref{thm:alg-ed} (with $a=O(e)$), we can implement this as an
$(1,O(e\log{n}),e/n,(e+1)/n)_{1,n^{-\Omega(e)}}$-\ttds.
Unfortunately, this algorithm
uses $\frac{n}{e}^{O(e)}$ time. Next, we show
how to obtain roughly the same space complexity but with
a much better time complexity.

\begin{theorem}
\label{thm:alg-tt}
Let $e,k,n\ge 1$ be integers such that $k\le n$ and $e\le n-k$. Then
\begin{itemize}
\item[(a)] If $e(e+k)\le O(n)$, then there exists a $(1,O(e\log^2{n}),e/n,2e/n)_{1,n^{-\Omega(1)}}$-\ttds\
for $\rs_S[k]$ under worst-case errors.
\item[(b)] If $ek\le O(n)$, then there exists a $(1,O(e\log^2{n}),e/n,2e/n)_{1,n^{-\Omega(1)}}$-\ttds\
for $\rs_S[k]$ under random errors.
\item[(c)] If $e\le O(\sqrt[s+1]{sn/k})$, then there exists a $(1,O(e^2\log^3{n}),e/n,5e/n)_{1,n^{-\Omega(1)}}$-\ttds\
for the folded RS code with folding parameter $s$ under worst-case errors.
\end{itemize}
Further, all the algorithms can be implemented in $\tilde{O}(n^2)$ time.
\end{theorem}

In the above, the soundness parameter follows by picking $a$ to be large 
enough while using Theorem~\ref{thm:alg-ed}. 
We observe that a naive implementation achieves the $\tilde{O}(n^2)$ runtime.
We also show that for part (a) and (b) by replacing $n$ by $n/\log{n}$ in the RHS of
the upper bound on $k$ and bumping up the space to $O(e^2\log^2{n})$, the algorithms
for $\rs_{\F_q}[k]$ can be implemented in $\tilde{O}(ne)$ time.
In fact, along with the faster running time, we get  $(1,O(e\log^2{n}),e/n,(e+1)/n)_{1,n^{-\Omega(1)}}$-\ttds s.

We start with some notation. Given an $t\times n$ (list) disjunct matrix
$M$
let $s$ and $s'$ denote the minimum and maximum Hamming weight of any row
in $M$. Further, let $D(N)$ denote the runtime of error detector for
$\rs_{\alpha_1,\dots,\alpha_N}[k]$.

We begin by analyzing the runtime of the tolerant tester from part (a)
of Theorem~\ref{thm:alg-tt}. The runtime has two parts: one is the time
taken to run the \edds\ for all the projected down codes, which are determined by the
rows of the $(e,e)$-list disjunct matrix $M$. Note that this step takes time
at most $t\cdot D(s')$. 
The second part is the time taken to run the algorithm from Lemma~\ref{lem:weird-ed},
which can be verified to be $D(n)$. Thus, the overall
running time is
\begin{equation}
\label{eq:total-time}
t\cdot D(s')+D(n).
\end{equation}
It can be verified that for list disjunct matrices from
Section~\ref{app:list-disj}, $t=O(e\log{n})$ and both $s$ and $s'$ are $\Theta(n/e)$.
Further, by Corollary~\ref{cor:detect-time}, we upper bound $D(N)$ by $\tilde{O}(N^2)$.
Thus, (\ref{eq:total-time}) implies that the runtime is upper bounded by
$\tilde{O}(n^2)$.

Next, we look at part (a) when $S=\F_q$. In this case we will pick an
\textit{explicit} disjunct matrix. This classic matrix is defined by
associating the columns with the codewords of $\rs_{\F_{q'}}[k']$ for
appropriate choices of $q'$ and $k'$. (Note that we then have $q=n=(q')^{k'}$.)
The columns of the matrix are the corresponding RS codewords, where each symbol
from $\F_{q'}$ in the codeword is replaced by the binary vector from $\{0,1\}^{q'}$,
which has  a $1$ only in the position corresponding to the symbol (when thought of as
an element from $[q']$). It is well known that if one picks $k'=q'/e$, then the matrix is
$e$-disjunct and $t=O(e^2\log^2{n})$ and $s=s'=n/\sqrt{t}$~\cite{test-book}. Note that one
can index the rows of this matrix by the tuples $(a,b)\in(\F_{q'})^2$.
For the rest of the argument fix such a row $(a,b)$. The columns that 
participate in this row correspond to the messages $(m_0,\dots, m_{k'-1})\in\F_{q'}^{k'}$
such that $\sum_{i=0}^{k'-1} m_i a^i=b$. Call these set of vectors
$S_{b}$. Before we proceed we recall that since $q=(q')^{k'}$, there is an isomorphism
between $\F_q$ and $\F_{q'}^{k'}$. Now note that $S_b$ is a linear subspace and thus, for
any $\gamma\in S_b$, $\prod_{c\in S_b, c\neq \gamma} (c-\gamma)$ (where we think of the operations
as happening over $\F_q$) is just the product of non-zero vectors in $S_b$, which is some
fixed constant (say) $\beta\in\F_q^*$. Now note that the error detection corresponding to
row $(a,b)$ is for the projected down code $\rs_{S_b}[k]$. Thus, we now satisfy the
second condition in Corollary~\ref{cor:detect-time}, which implies that we can assume that
the error detection can be done in linear time. Finally, for $\rs_{\F_q}[k]$ it is well known
that the second condition in Corollary~\ref{cor:detect-time} is also satisfied. Thus the overall
runtime is bounded by $\tilde{O}(t\cdot n/\sqrt{t}+n)$, which is $\tilde{O}(ne)$, as
desired.

(We remark that if we can get the best of both the random list disjunct matrix construction,
i.e. $t=O(e\log{n})$ and both $s,s'$ in $\tilde{\Theta}(n/e)$, and the explicit RS code
based disjunct matrix, i.e. the second condition of Corollary~\ref{cor:detect-time} is true, then
we can have a tolerant tester with the optimal runtime of $\tilde{O}(n)$.)

The proof for the runtime for part (b) in Theorem~\ref{thm:alg-tt} is identical and is
omitted. The proof for the naive implementation runtime for part (c) is similar to part
(a)-- everything gets multiplied by $O(\log{e})$ factor, which  is at most an extra
log factor. We do not know of an explicit (list) disjunct matrix that satisfies the 
extra requirements for part (c) and thus, we do not have any implementation with runtime
better than $\tilde{O}(n^2)$.

For the rest of the section, we will focus on the other parameters of the
algorithms.

All of the results above follow from a generic reduction that uses
group testing. In particular, 
let $\vy$ be the received word that we wish to test, and $\vc$ be the nearest
codeword to $\vy$.  Let $\vx\in \{0,1\}^n$ be the characteristic vector
associated with error locations in $\vy$ with respect to $\vc$. The high level idea
is essentially to figure out $\vx$ using group testing.

Let $M$ be a $t\times n$ binary matrix that is $(e,e)$-list disjunct. By Section~\ref{app:list-disj} we can get our hands on $M$ with $t=O(e\log{n})$ with
$O(e\log^2{n})$ space. Now consider the following natural algorithm.

\begin{quote}
For all $i\in [t]$, check if $\vy_{M_i}\in RS_{M_i}[k]$, where $M_i$ is the subset
corresponding to the $i$th row of $M$.  If so, set $r_i=0$, else
set $r_i=1$.  Run $\mathcal A$ from Proposition~\ref{prop:new-disj}
with $\vr$ as input, to get $\vxhat$.  (\textsc{Step 1}) If
$\wt(\vxhat)\ge 2e$, declare that $\ge 2e$ errors have occurred.
(\textsc{Step 2}) If not, declare $\le e$ errors iff $\vy_{S\setminus T}\in\rs_{S\setminus T}[k]$, where $T$ is the subset corresponding to
$\vxhat$.
\end{quote}

The way the algorithm is stated above, it seems to require two passes. However, using
Lemma~\ref{lem:weird-ed}, we can run \textsc{Step 2} in parallel with the rest of the
algorithm, resulting in a one-pass implementation.

Let $\vz$ be the result of applying $M$ on $\vx$. Now if it is the case that
$z_i=1$ iff $r_i=1$, then the correctness of the algorithm above follows
from the fact that $M$ is $(e,e)$-list disjunct and Proposition~\ref{prop:new-disj}.
(If $\wt(\vx)\le e$, then we have $S_{\vx}\subseteq S_{\vxhat}$ 
(where $S_{\vx}$ is the subset of $[n]$ whose incidence vector is $\vx$)
and $\wt(\vxhat)<2e$,
in which case the algorithm will declare at most $e$ errors. Otherwise, the
algorithm will ``catch" the at least $e$ errors in either \textsc{Step 1} or failing which, in
\textsc{Step 2}.) 

However, what complicates the analysis is the fact that even though $z_i=0$
implies $r_i=0$, the other direction is not true. In particular, we could have
$\vy_{M_i}\in\rs_{M_i}[k]$, even though $(\vy-\vc)_{M_i}\neq \vzero$. The three parts
of Theorem~\ref{thm:alg-tt} follow from different ways of resolving this problem.

Note that if $\wt(\vx)\ge 2e$, then we will always catch it in \textsc{Step 2}
in the worst-case. So from now on, we will assume that $0<\wt(\vx)\le e$.
Let the minimum support of any row in $M$ be $s$.

We begin with part (a).
Let $s>k+e$ and define $\Delta=\vy-\vc$. 
Note that we are in the case where
$0<$ \textsc{wt}$(\Delta)\le e$.  Since $s\ge k$ and $\vc_{M_i}\in RS_{M_i}[k]$,
$\vy_{M_i}\in RS_{M_i}[k]$ if and only if $\Delta_{M_i}\in RS_{M_i}[k]$.
Note also that for any $i$, \textsc{wt}$(\Delta_{M_i})\le $\textsc{wt}$(\Delta)\le e$.
Now, the distance of $RS_{M_i}[k]$ is $s-k-1 > e$, so for every $i$
with non-zero $\Delta_{M_i}$,
$\Delta_{M_i}\not\in\rs_{M_i}[k]$, which in turn means that $z_i=1$ will always imply
that $r_i=1$ when $M$ has the stated support.  By Section~\ref{app:list-disj}, we have $s\ge n/(2e)$, which concludes the proof
of part (a).

The following lemma follows from the random errors result in~\cite{two-theorems} and is needed for part (b):

\begin{lemma}[\cite{two-theorems}]
\label{lem:rs-avg-list-one}
Let $k\le n< q$ be integers such that $q>\left(\frac{n}{k}\right)^2$. Then
the following property holds for RS codes of dimension $k$ and
block length $n$ over $\F_q$ : For $\ge 1-q^{-\Omega(k)}$ fraction of
error patterns $\ve$ with $\wt(\ve)\le n-4k$ and any codeword
$\vc$, the only codeword that agrees in at least $4k$ positions with
$\vc+\ve$ is $\vc$.
\end{lemma}

Now if $s\ge 4k$, then with high probability, every non-zero
$\Delta_{M_i}\not\in\rs_{M_i}[k]$ (where $\Delta$ is as defined
in the proof of part (a)).
The fact that $s\ge n/(2e)$ completes the proof of part (b).

The proof of part (c) is more involved and needs a strong connection to
the list decodability of the code being tested, which we discuss next.

\vspace*{2mm}
\noindent
\textbf{Connection to List Decoding.}
%
Unlike the proofs of part (a) and (b) where the plain vanilla
$(e,e)$-list disjunct matrix works, for part (c), we need and
use a stronger notion of list disjunct matrices.  We show that if the list disjunct
matrix is picked at random, the bad tests (i.e. $r_i=0$ even though $z_i=1$)
do not happen often and thus, one
can decode the result vector even with these errors.  We show that these kind of
matrices suffice as long as the code being tested
has good enough list decodability. The tolerant testing algorithm for a Reed-Solomon code,
for instance, recursively reduces the amount of errors that need to be detected, and after
application of Lemma~\ref{lem:weird-ed}, can be made to accomplish this in a single pass.  We
also show that the relevant list disjunct matrices can be found, with high probability, using low space
and a low number of random bits.

We need to show what the forbidden subsets will be in our
setting. Let $C$ be the code we are trying to test.
Let $\vy$ be the received word and let $\vc\in C$ be 
such that $\Delta(\vy,\vc)\le e\le d/2$, where $d$ is the distance of
$C$. Let $C$ be $(n-a,L+1)$-list decodable, that is, for any Hamming ball
of radius at most $n-a$, there are at most $L+1$ codewords from $C$ in it.
Note that if $n-a\ge \Delta(\vy,\vc)$, then there are at most $L$
codewords (other than $\vc$) that agree with $\vy$ in at least $a$
positions. Also note that each such codeword agrees with  $\vy$
in at most $n-d/2$ positions. Let $T\subseteq [n]$ be the set of
positions where $\vy$ and $\vc$ agree. Then define $\mathcal{F}_{a,n-d/2}(T)$
to be the (at least $a$) positions where codewords other than $\vc$ agree with $\vy$. As 
$C$ is $(n-a,L+1)$-list decodable, $|\mathcal{F}_{a,n-d/2}(T)|\le L$.

We now show how one can use list disjunct matrices from Definition~\ref{def:new-disj} to construct
data stream algorithms for tolerant testing of
an RS code
$C$. 
Assume that there exists a $0\le \g\le 1$, such that for every (large enough)
$f\ge 1$, there exists a strongly explicit $(f,f,\g,\mathcal{F}_{a,n-d/2})$-list disjunct
matrix $M_f$. We next show how these matrices can be
used to solve the tolerant testing problem for $C$ where we want to
distinguish between the case that at most $e$ errors have occurred and
at least $2e\frac{2-\g}{1-\g}+1$ errors have occurred.

First consider the case when $e\le k$. In this case we can use the 
algorithm from part (a) of Theorem~\ref{thm:alg-tt}. Thus, if $e\le k$, then we
can handle $e$ vs. $2e$ errors, where $e\le n/(2k)$, in space $O(e\log^2{n})$.
Now consider the case when $e>k$.
Let us use the matrix $M_e$ on $\vy$ as we did before
That is, for every row of $M_e$ (in particular,
the corresponding subsets $S\subseteq [n]$), check if $\vy_S\in RS_S[k]$.
If so, assign $r_i=0$ (otherwise assign $r_i=1$). Given the result
vector $\vr$, run $\mathcal{A}$ from Proposition~\ref{prop:new-disj} on it
to obtain a subset $G\subseteq [n]$ such that $|G|\le 2e$ and $[n]\setminus
G$ contains at most $\g e$ errors. Thus, we have reduced the problem from at most
$e$ errors out of $n$ positions to the problem of at most $\g e$ errors in at least
$n-2e$ positions. Then, the rest is natural: recurse on this idea. We stop
when we are left with at most $k$ errors.
Note that we will need $O(\log{e}/\log(1/\g))$ many recursions. Because of
these recursions, we can handle the case when there are at most $e$ errors
vs at least $2e(1+\g +2\g^2 +\dots)+2k$ errors. This implies (as $e>k$)
we can definitely handle $e$ vs.  $2e\frac{2-\g}{1-\g}+1$ errors.
The way the idea is stated above, it seems like an $O(\log{e}/\log(1/\g))$-pass
algorithm. However, using Lemma~\ref{lem:weird-ed}, the 
algorithm outlined above can be implemented in one pass.

We are all done except the construction of the list disjunct matrices as
defined in Definition~\ref{def:new-disj}:

\begin{theorem}
\label{thm:new-disj-limited-rand}
Let $e,n,d,a,L\ge 1$ be integers such that $e\le O\left(\frac{d}{\log{L}}\right)$. There
exists a large enough constant $c>1$ such that if $c\cdot e\log{n}\le n$, then the following holds:
There exists a $\left(e,e,\frac{1}{60},\mathcal{F}_{a,n-d/2}\right)$-list disjunct matrix with
$t=ce\log{n}$ rows (where for every $E\subseteq [n]$ with $|E|\le e$, $|\mathcal{F}_{a,n-d.2}(E)|\le L$).
Further, every row has at least $\frac{n}{2e}$ ones in it.
In addition, one can construct such matrices with probability at least
$1-n^{-\Omega(1)}$ using $R=O(t(e+\log{L})\cdot\log{e}\cdot\log{n})$ random bits.
Further, given these $R$ bits, any entry of the matrix can be computed in
$\poly{\log{n}}$ space.
\end{theorem}

A folded
RS code (with ``folding parameter" $s$), is
$\left(n-\sqrt[s+1]{sk^sn}, n^{O(s)}\right)$-list
decodable~\cite{GR-capacity}. Thus, Theorem~\ref{thm:new-disj-limited-rand}
proves part (c) of Theorem~\ref{thm:alg-tt}.
(Note that Theorem~\ref{thm:new-disj-limited-rand}
also has the constraint that $e\le O(d/\log{L})$. However since $\log{L}$
is $O(\log{n})$ above and as long as $d=\Omega(n)$, this bound is much weaker.)

We prove the existence of the required object by the probabilistic method.
In fact this proves the second part but with $R=O(nt\log{e})$. To reduce
the randomness, we observe that the proof only requires bits that come
from an $O(t(e+\log{L})\log{e})$-wise independent source.

Let $M$ be a $t\times n$ matrix, where each entry is chosen to be one with
probability $1/e$ and $t=c\cdot e\log{n}$ for some large enough 
constant $c$.\footnote{In this proof, we have not attempted to optimize
the constants. By a conservative estimate, picking $c=10^7$ would
suffice.} 

We first argue about the minimum support size of any row in $M$. It is easy
to check that the expected Hamming weight of any row in $M$ is exactly
$n/e$. Thus, by the Chernoff bound, the probability that any row has
Hamming weight at most $n/2e$ is upper bounded by
\begin{equation}
\label{eq:support}
\exp\left(-\frac{n}{12e}\right)\le \exp(-c\log{n}/12)\le n^{-190},
\end{equation}
where the last inequality follows for large enough $c$. Now by the union bound
(and the fact that $t\le n$), all the rows have Hamming weight at least
$n/(2e)$ with probability at least $1-n^{-189}$.

Next we move on to proving property (a) from Definition~\ref{def:new-disj}
for $M$ with $b_1=\frac{t}{16e}$. To this end, fix subsets $U,T\subseteq [n]$
with $|U|=e,|T|= e$ and $U\cap T=\emptyset$. Call a row $j\in [t]$
\textit{good} if there exists a $i\in U$ such that $M_{j,i}=1$ and
$M_{j,\ell}=0$ for every $\ell\in T$. Now the probability that a row is
good is exactly
\begin{equation}
\label{eq:parta-1}
\left(1-\left(1-\frac{1}{e}\right)^e\right)\left(1-\frac{1}{e}\right)^e\ge \frac{1}{8},
\end{equation}
where the inequality follows if $e\ge 2$ and the fact that $(1-1/x)^x\le \exp(-1)\le 1/2$.
Thus, the expected number of good rows is at least $t/8$. By the Chernoff
bound, the probability that the number of good rows is at most $t/16$ is
upper bounded by
\begin{equation}
\label{eq:parta-2}
\exp\left(-\frac{t}{96}\right)= \exp\left(-\frac{ce\log{n}}{96}\right)\le n^{-190e},
\end{equation}
where the inequality follows for large enough $c$. Thus, with high probability,
the number of good rows is at least $t/16$. Then by the pigeonhole
principle, at least one column $i\in U$ is contained in at least
$\frac{t}{16e}\stackrel{def}{=}b_2$ good rows. Taking the union bound over the
$\binom{n}{e}\binom{n-e}{e}$ choices of $T$ and $U$ implies that with probability
at least $1-n^{-188e}$, property (a) is satisfied for every valid choice
of $T$ and $U$.

Next, we move on to the more involved part of the proof, which is to prove
property (b) in Definition~\ref{def:new-disj}. To this end, given any
$T\subseteq [n]$ with $|T|\le e$ and a column $i\in [n]$ we will upper
bound the probability that at least $b_1$ tests that contain $i$ are themselves
contained in some subset in $\mathcal{F}_{a,n-d/2}(T)$. It turns out that this
probability will be $n^{-O(1)}$, which is not enough to apply the union bound
over all the $\binom{n}{e}$ choices of $T$. We then observe that these 
probabilities are almost independent for any $\Omega(e)$ such columns, which
is sufficient for the union bound over all choice of $T$ to go through.

Fix a subset $T\subseteq [n]$ with $|T|\le e$ and a subset $S\in\mathcal{F}_{a,n-d/2}(T)$. (Note that $|S|\le n-d/2$.)
We say that a row $j\in [t]$ \textit{avoids} $S$ ($\mathcal{F}_{a,n-d/2}(T)$ resp.) if the $j$th row (which we
will denote by $M(j)$) is not a subset of $S$ (any subset in $\mathcal{F}_{a,n-d/2}(T)$ resp.). 
In other words, if $j$
avoids $S$ then $M_{j,i}=1$ for some $i\not\in S$. Thus, we have
\begin{equation}
\label{eq:avoid-1}
\prob{j\text{ doesn't avoid } S}=\left(1-\frac{1}{e}\right)^{n-|S|}
\le \left(1-\frac{1}{e}\right)^{-d/2}\le \exp\left(-\frac{d}{2e}\right),
\end{equation}
where first inequality follows from the fact that $|S|\le n-d/2$.

Fix a column $i\in [t]$. Note that if $i\not\in S$ and $M(j)$ contains $i$,
then $j$ does avoid $S$. Now, if $i\in S$ and given that the
probability in (\ref{eq:avoid-1}) only depends on the indices $i\not\in S$, we
get that $\prob{j \text{ contains } i \text{ and doesn't avoid } S}\le \exp(-d/(2e))/e$.
Thus, whether $i\in S$ or not, we have by the union bound
\begin{equation}
\label{eq:bad-row}
\prob{j\text{ contains } i \text{ and doesn't avoid }\mathcal{F}_{a,n-d/2}(T)}
\le \frac{L}{e}\cdot\exp\left(-\frac{d}{2e}\right)\le \frac{1}{80e},
\end{equation}
where the last inequality follows if $e\le \frac{d}{10\ln{L}}$.

Now call a row $j$ $i$-\textit{bad} if it contains $i$ but does not avoid
$\mathcal{F}_{a,n-d/2}(T)$. (If it contains $i$ and avoids $\mathcal{F}_{a,n-d/2}(T)$, the call it $i$-good.)
 Note that we need to show that for at least
$(1-\g)d$ columns $i\in T$, there are at most $b_1$ $i$-bad rows.
Thus, by (\ref{eq:bad-row}), the expected number of $i$-bad rows is
at most $t/(80e)$, or the expected number of $i$-good rows is at least
$79t/(80e)$. By the Chernoff bound, we have
\begin{equation}
\label{eq:avoid-2}
\prob{\text{Number of } i\text{-good rows}< \frac{78t}{80e}=\frac{39t}{40e}}\le
\exp\left(-\frac{t}{3\cdot 79\cdot 80e}\right)\le n^{-190},
\end{equation}
where the last inequality follows from large enough $c$. Since the expected
Hamming weight of any column is $t/e$, the Chernoff bound
implies that
\begin{equation}
\label{eq:avoid-3}
\prob{\text{Column } i \text{ has Hamming weight }\ge \frac{81t}{80e}}\le \exp\left(-\frac{t}{3\cdot 80^2 e}\right)\le n^{-190}, 
\end{equation}
where again the last inequality follows for large enough $c$. Thus, (\ref{eq:avoid-2})
and (\ref{eq:avoid-3}) imply that
\begin{equation}
\label{eq:num-bad}
\prob{\text{Number of } i\text{-bad rows}> \frac{3t}{80e}}\le 2\cdot n^{-190}
\le n^{-189},
\end{equation}
where the last inequality is true for $n\ge 2$.
Unfortunately, the bound above is too weak to apply the union bound over all
the $\binom{n}{e}$ choices of $T$. However, we get around this obstacle by 
proving that for any $\Omega(e)$ values of $i\in [t]$, the probabilities
above are essentially independent. 

Call a column $i\in T$ \textit{bad} if the number of bad $i$-rows is
at least $b_1$ (for some
$\frac{3t}{80e}<b_1<\frac{t}{16e}$ to be fixed later). 
For notational convenience, define
$\ell=\frac{e}{60}$. Next we are going to show that for any subset
$V=\{i_1,\dots,i_{\ell}\}\subseteq T$,
\begin{equation}
\label{eq:indep}
\prob{\text{Every } j\in V\text{ is bad}}\le n^{-3e}.
\end{equation}
Note that the above implies that the probability that more than $\ell$
columns in $T$ are bad is upper bounded by
\[\binom{e}{\frac{e}{60}} n^{-3e}\le n^{-2e},\]
where the last inequality follows for $n\ge 180$. Thus, the probability that
for some $T\subseteq [n]$ with $|T|\le e$, there are more than $\ell$ bad
columns in $T$, by the union bound, is upper bounded by $n^{-e}$. 
Thus, property (b) is true with probability at least $1-n^{-e}$.

To complete the proof, we will prove (\ref{eq:indep}). Note that we can rewrite
the probability in (\ref{eq:indep}) as
\[\prob{i_{\ell}\text{ is bad}\mid\bigwedge_{j\in V\setminus i_{\ell}}j\text{ is bad}}\cdot\prod_{j\in V\setminus i_{\ell}}\prob{j\text{ is bad}}
\le \prob{i_{\ell}\text{ is bad}\mid\bigwedge_{j\in V\setminus i_{\ell}}j\text{ is bad}}\cdot n^{-2e},\]
where the inequality follows from (\ref{eq:num-bad}) and the fact that
$b_1>3t/(80e)$. Thus, we will be
done if we can show
\begin{equation}
\label{eq:to-show}
\prob{i_{\ell}\text{ is bad}\mid\bigwedge_{j\in V\setminus i_{\ell}}j\text{ is bad}}\le n^{-e}.
\end{equation}
To this end, let $B\subseteq [t]$ be the rows that contain at least one
column from $V\setminus i_{\ell}$. 
Note that
\begin{equation}
\label{eq:avoid-3-5}
\av[|B|]=t\left(1-\left(1-\frac{1}{e}\right)^{\ell}\right).
\end{equation}
By the Chernoff bound, we have
\begin{equation}
\label{eq:avoid-4}
\prob{|B|\ge \frac{6t}{5}\left(1-\left(1-\frac{1}{e}\right)^{\ell}\right)}\le \exp\left(-\frac{t}{75}\cdot \left(1-\left(1-\frac{1}{e}\right)^{\ell}\right)\right).
\end{equation}
As for any real $x>0$ and integer $y>0$ with $xy<1$, $1-xy\le (1-x)^y\le 1-xy+(xy)^2/2$, we have $59/60\le \left(1-\frac{1}{e}\right)^{\ell} \le 59/60+1/7200<119/120$. This along with (\ref{eq:avoid-4}) implies that
\[\prob{|B|>\frac{t}{50}}\le \exp\left(-\frac{t}{75\cdot 120}\right)\le n^{-190e},\]
where the last inequality follows for large enough $c$. Similarly, one can show
that
\[\prob{\text{Number of rows in } B \text{ that contain } i_{\ell}>\frac{t}{50e}}\le n^{-190}.\]
We do a conservative estimate and  assume that all tests in $B$ that contain
$i_{\ell}$ are $i_{\ell}$-bad. Because of the bound above, w.l.o.g.
with all but an $n^{-190}$ probability, we can assume that $|B|=t/50$.
Using the same calculation\footnote{We need to
replace $t$ by $49t/50$. Further in (\ref{eq:avoid-1}), we need to replace
$n-|S|$ by $n-|S|-e/60$ as in the worst case $\{i_1,\dots,i_{\ell-1}\}\subseteq [n]\setminus S$. However, this does not change the upper bound in (\ref{eq:bad-row}) as long as we pick $e\le d/(15\ln{L})$.} as we did to
obtain (\ref{eq:num-bad}), we can show that 
\[
\prob{\text{Number of } i_{\ell}\text{-bad rows in }B>\frac{3\cdot 49t}{4000e}}\le 2\cdot n^{-49\cdot 190/50}\le n^{-185}.\]
Adding in the number of rows in $B$ that contain  $i_{\ell}$, we obtain that
\[\prob{\text{Number of } i_{\ell}\text{-bad rows} >\frac{t}{50e}+\frac{147t}{4000e} \mid \bigwedge_{j\in V\setminus i_{\ell}} j\text{ is bad}}\le n^{-190}+n^{-185}\le n^{-180}.\]
Picking $b_1=\frac{t}{e}\left(\frac{1}{50}+\frac{147}{4000}\right)<b_2$ completes the
proof of (\ref{eq:to-show}). Thus, we have completed the proof of the
existence of the desired $(e,e,\frac{1}{60},\mathcal{F}_{a,n-d/2})$-list
disjunct matrix.

In fact, the proof shows that the required matrices can be computed with
high probability. However, at least $\Omega(n)$ random bits are required,
which is too high for any data stream application. Next we point out that
the proof only requires limited independence and hence, we can get
away with much fewer random bits. In the remainder of the proof, we will think
of the bits of $M$ to come from some $k$-wise independent source that
contain bit strings of length $tn$.

We now go through the proof above and estimate the amount
of independence needed. The first place that needs independence
is (\ref{eq:support}) and we claim
that $O(\log{n})$-wise independence suffices. This follows from the
tail bounds for $k$-wise independent sources from~\cite{limited-indep-chernoff}.
In particular, Bellare and Rompel show that for a $k$-wise independent source,
the sum of binary random variables with mean $\mu$ can have a deviation
of strictly more than $A$ with probability at most $8\cdot\left(\frac{\mu k+k^2}{A^2}\right)^{k/2}$. Note that in our case $\mu=n/e$, $A=n/(2e)$ and since
$n/e\ge c\log{n}$, picking a $O(\log{n})$-wise independent source works.

The next places  in the proof that use independence are (\ref{eq:parta-1}) and (\ref{eq:parta-2}).
It is easy to check that the calculations go through if we have $2et$-wise
independence. Next, independence is used in (\ref{eq:avoid-1}). Note
that in this case we need to upper bound the probability by
$(1-1/e)^{\Omega(e\log{L})}$. Thus, picking  $O(e\log{L})$-wise
independence works for this case. Next
(\ref{eq:avoid-2}) needs $t/e$-wise independence. This follows from the
tail bound for $k$-wise independence from~\cite{limited-indep-chernoff}. Note
that we actually need the product of the independence used in (\ref{eq:avoid-1})
and (\ref{eq:avoid-2}), that is, we need a total of $O(t\log{L})$-wise
independence.
For (\ref{eq:avoid-3}) $t$-wise independence suffices.
Finally for (\ref{eq:avoid-3-5}) and (\ref{eq:avoid-4}) we need $\ell t$-wise
independence. In fact, again using the bound from~\cite{limited-indep-chernoff},
we can get away with $O(t)$-wise independence. 

Thus, overall we need $O(t(e+\log{L}))$-wise independence. Generally,
$k$-wise independent sources are for unbiased bits where as we need random
bits that take a value of one with probability $1/e$. However, since we
can get such random bits from $O(\log{e})$ unbiased bits, we will need
$O(t(e+\log{L})\log{e})$-wise independent sources containing $O(nt\log{e})$
bit strings. Using well-known construction of $k$-wise independent sources,
we can get away with $R=O(t(e+\log{L})\log{e}\cdot \log(nt))$ random bits.
This completes the proof as $t\le n$.

\subsection{List Disjunct Matrices.}
We begin with the definition of a stronger kind of list disjunct matrices:

\begin{definition}
\label{def:new-disj}
Let $n,s_1,s_2,e,\ell,L \ge 1$ be integers with $s_1\le s_2$ and
let $0\le \g\le 1$ be a real. For any subset $T\subset [n]$ such that
$|T|\le e$, let $\mathcal{F}_{s_1,s_2}(T)$ be a collection of {\em forbidden} subsets
of $[n]$ of size in the range $[s_1,s_2]$ such that $|\mathcal{F}_{s_1,s_2}(T)|\le L$. A $t\times n$
binary matrix $M$ is called a $(e,\ell,\g,\mathcal{F}_{s_1,s_2})$-{\em list disjunct} matrix if there exist
integers $0\le b_1< b_2$ such that the following
hold for any $T\subseteq [n]$ with $|T|\le e$:
\begin{enumerate}
\item {For any  subset $U\subseteq [n]$ such that $|U|\ge \ell$ and $U\cap T=\emptyset$, there exists an $i\in U$ with the
following property: The number of rows where the $i$th column of $M$ has a one and all the columns in $T$ have a zero is at
least $b_2$.}
\item {The following holds for at least $(1-\g)e$ many $i\in T$: Let $R_i$ denote all the rows of $M$ (thought of as subsets
of $[n]$) that contain $i$. Then $|\{U\in R_i|U\subseteq V, \text{ for some } V\in\mathcal{F}_{s_1,s_2}(T)\}|\le b_1$.}
\end{enumerate}
\end{definition}
The definition might appear complicated but it is setup to easily imply 
Proposition~\ref{prop:new-disj}. Further, a $(e,\ell,0,\emptyset)$-list
disjunct matrix (with $b_1=0$ and $b_2=1$) is the same as the $(e,\ell)$-list disjunct matrix considered
in~\cite{list-disj}. Further, an $(e,1)$-list disjunct matrix is the well-known
$e$-disjunct matrix~\cite{test-book}.

Let us also define the following error version of group testing that will be
relevant to our scenario. 

\begin{definition}
\label{def:new-group}
Let $n,s_1,s_2,e,L\ge 1$ be integers. For every $T\subseteq [n]$ such
that $|T|\le e$, let $\mathcal{F}_{s_1,s_2}(T)$
be the collection of forbidden subsets as in Definition~\ref{def:new-disj}.
Then $(e,\mathcal{F}_{s_1,s_2})$-{\em group testing} works in the following manner:
Given a set of defectives $T\subseteq
[n]$ such that $|T|\le e$, any test $U\subseteq [n]$ behaves as follows:
If $U\cap T=\emptyset$, then the test will return an answer of $0$. If
$U\cap T\neq\emptyset$ and $U\subseteq V$ for some $V\in \mathcal{F}_{s_1,s_2}(T)$,
then the test will return an answer of $0$. 
Otherwise the test returns
an answer of $1$.
\end{definition}

The algorithm $\mathcal A$ in the below proposition is a natural generalization
of the standard decoding algorithm for $e$-disjunct
matrices~\cite{test-book}.

\begin{prop}
\label{prop:new-disj}
Let $n,e,\ell,s_1,s_2,L,\g,\mathcal{F}_{s_1,s_2}$ be as in Definition~\ref{def:new-disj}. Let
$M$ be a $(e,\ell,\g,\mathcal{F}_{s_1,s_2})$- list disjunct
matrix with $t$ rows. 
Finally, consider an outcome vector
$\vr$ of applying $M$ to a set of defectives $E$ with $|E|\le e$
in the $(e,\mathcal{F}_{s_1,s_2})$-group testing scenario. Then there exists an
algorithm $\mathcal{A}$, which given $\vr$ can compute a set $G$ such that
$|G|\le \ell+e-1$ and $|E\setminus G|\le \g e$. Further,
$\mathcal{A}$ uses $O(t+\log{n}+ S(t,n))$ space, where $S(t,n)$ is the 
space required to compute any entry of $M$. 
\end{prop}

\subsection{Proof of Proposition~\ref{prop:new-disj}}
The algorithm $\mathcal{A}$ is very simple: Go through every column $i\in [n]$
and declare $i\not\in G$ if and only if the number of rows $j\in [t]$ where
$M_{j,i}=1$ but $r_j=0$ is at least $b_2$. It is easy to check that
$\mathcal{A}$ has the claimed space requirement. The correctness of $\mathcal{A}$
follows from Definitions~\ref{def:new-disj} and~\ref{def:new-group}.
To see this note that if
$|G|\ge e+\ell$, i.e. $|G\setminus E|\ge \ell$, then by part (a) of Definition
\ref{def:new-disj},
there exists an $i\in G\setminus E$ with the following property:
There are at least $b_2$
rows $j\in[t]$ such that $M_{j,i}=1$ but $M_{j,i'}=0$ for every $i'\in E$.
By Definition~\ref{def:new-group}, for every such $j$, $r_j=0$. Thus, by definition
of $\mathcal{A}$, $i$ cannot be in $G$.
Now consider an $i\in E$ for which property (b) of Definition~\ref{def:new-disj} holds. 
Now by
Definition~\ref{def:new-group}, there are at most $b_1$ rows $j\in [t]$
such that $M_{j,i}=1$ and $r_j=0$. Since $b_1<b_2$, $\mathcal{A}$ includes
$i$ in $G$. This implies that $|E\setminus G|\le \g e$.



The space requirement of $O(e^2\log^2{n})$ of part (c)
is unsatisfactory.
Reducing the amount of randomness
needed to something like $O(e\log{n})$ will realize the full potential our 
algorithm. We leave this as an open problem.

\section{Limitations of our techniques}
\label{sec:limits}

One shortcoming of Theorem~\ref{thm:alg-tt}
is that to distinguish between (say) at most $e$ and at least $2e$
errors, we needed $e\cdot s\le O(n)$, where $s$ is the minimum support
size of any test. 
Another shortcoming is that we need $O(e\log{n})$ space.
In this section, we prove that our techniques cannot
overcome these limits.

We begin with some quick notation. For any $k\ge 1$, a
$k^{++}$ query to a string $x\in \{0,1\}^n$ corresponds to a subset
$S\subseteq [n]$. The answer to the query is $x_S$ if $\wt(x_S)<k$, otherwise
the answer is $k^{++}$ (signifying that $\wt(x_S)\ge k$). (This is a 
natural generalization of $k^+$ decision trees considered by Aspnes et al.
\cite{k-plus}.)
A $k^{++}$ algorithm
to solve the $(\ell,t,n)$-threshold function makes a sequence of $k^{++}$ queries to
the input $x\in \{0,1\}^n$, and can tell whether $\wt(x)\le \ell$ or
$\wt(x)\ge t$. 
If we think of $x$ as being the indicator vector for error locations, then our
reduction from tolerant testing to error detection can be thought of as a
$1^{++}$ algorithm for the $(e,O(e))$-threshold function. 

First we show that the minimum support size that we obtain in our reduction,
even with the stronger $k^{++}$ primitive, is nearly optimal.

\begin{theorem}
\label{thm:support-lb}
Let $0\le \ell < t\le n$ and $k\ge 1$ be integers. Let $\eps<1/2$ be a constant
real.
Then any non-adaptive, randomized $k^{++}$ algorithm for the $(\ell,t,n)$-threshold
problem with error probability at most $\eps$, where all the queries have support size at least $s$, needs to make
at least 
$e^{s\ell/n}/n^{O(k)}$
queries. In particular, any algorithm that makes a sublinear number of queries
needs to satisfy
$ s\cdot \ell\le O(kn\log{n})$.
\end{theorem}

\subsection{Proof of Theorem~\ref{thm:support-lb}}

Define the following distribution $\mathcal{D}$ on inputs in $\{0,1\}^n$: uniformly distribute
a probability mass of $1/2$ over the $\binom{n}{\ell}$ vectors of
Hamming weight exactly $\ell$ (call this set $\mathcal{N}$) and the rest of the probability mass uniformly
over the $\binom{n}{t}$ vectors of Hamming weight $t$ (call this set $\mathcal{Y}$). We will show that any
\textit{deterministic} non-adaptive $k^{++}$ algorithm with an
error probability at most $\eps$ (according to $\mathcal{D}$)
must make at least $\frac{e^{\frac{s\ell}{n}}}{n^{O(k)}}$
queries. Yao's lemma will then complete the proof.

Fix an arbitrary $k^{++}$ algorithm $A$ that has error probability at most $\eps$. Thus, $A$ outputs the correct value of
$0$ in at least 
$\frac{1/2}{(1/2-\eps)}\ge 1-2\eps$ 
fraction of elements in $\mathcal{N}$ 
(call this set of elements $\mathcal{N}_0$). Similarly, the algorithm outputs
the correct value of $1$ in at least $1-2\eps$ fraction of the elements in
$\mathcal{Y}$ (call this set $\mathcal{Y}_1$). Any $k^{++}$ query is
said to \textit{cover} a pair of inputs
$(x_0,x_1)\in\mathcal{N}_0\times \mathcal{Y}_1$, if it outputs different
answers for the inputs $x_0$ and $x_1$. Note that 
all the pairs in $\mathcal{N}_0\times \mathcal{Y}_1$ have to be covered by some
query in $A$.

To complete the proof, we will show that at least $e^{s\ell/n}/n^{O(k)}$ queries 
are needed to
cover $\mathcal{N}_0\times \mathcal{Y}_1$. To this end given an arbitrary
query $Q$ of support $i\ge s$, we will bound the number of pairs it can
cover (call this number of pairs $P_Q$). Note that $Q$ will not cover a pair $(x_0,x_1)$ if both $x_0$ and $x_1$ have
at least $k$ ones in the support of $Q$. Thus, to upper bound $P_Q$, we will 
count the number of pairs $(x_0,x_1)$ such
that either $x_0$ or $x_1$ have support $<k$ in $Q$. This latter count is clearly
upper bounded by 
\[\max\left(\binom{n}{\ell}\cdot\left(\sum_{j=0}^{k-1} \binom{i}{j}\binom{t}{j}\binom{n-i}{t-j}\right), \left(\sum_{j=0}^{k-1} \binom{i}{j}\binom{\ell}{j}\binom{n-i}{\ell-j}\right)\binom{n}{t}
\right),\]
where for notational convenience we define $\binom{a}{b}=0$ for $b>a$.
We claim that the above is upper bounded by
(see Appendix~\ref{app:sum-ub} for a proof):
\[kn^{3(k-1)}\cdot \max\left(\binom{n}{\ell}\binom{n-s}{t'},\binom{n-s}{\ell'}\binom{n}{t}\right),\]
where $t'=\max_{0\le j\le k-1}\{t-j| t-j\le n-s\}$ and $\ell'=\max_{0\le j\le k-1}\{\ell-j|\ell-j\le n-s\}$.
The way we are going to proceed with the rest of the proof, the maximum
in the above will occur for the second argument, i.e. from now on,
we have that for any query $Q$,
\begin{equation}
\label{eq:p-max}
P_Q\le kn^{3(k-1)}\binom{n-s}{\ell'}\binom{n}{t}\stackrel{def}{=}P_{max}.
\end{equation}
As
\begin{equation}
\label{eq:universe}
|\mathcal{N}_0\times \mathcal{Y}_1|\ge (1-2\eps)^2\binom{n}{\ell}\binom{n}{t},
\end{equation}
by the pigeonhole principle, the number of queries that $A$ needs to make
is at least
\begin{align}
\label{eq:set-cover-1}
\frac{|\mathcal{N}_0\times \mathcal{Y}_1|}{P_{max}}&\ge
\frac{ (1-2\eps)^2\binom{n}{\ell}\binom{n}{t}}{kn^{3(k-1)}\binom{n-s}{\ell'}\binom{n}{t}}\\
\label{eq:set-cover-2}
&\ge \frac{ (1-2\eps)^2\binom{n}{\ell'}}{kn^{4(k-1)}\binom{n-s}{\ell'}}\\
\label{eq:set-cover-3}
&\ge \frac{ \sqrt{8}(1-2\eps)^2e^{s\ell'/n}}{kn^{4(k-1)}\sqrt{27(n+1)}}\\
\label{eq:set-cover-4}
&\ge \frac{ \sqrt{8}(1-2\eps)^2e^{s\ell/n}}{kn^{4(k-1)}e^{k-1}\sqrt{27(n+1)}}.
\end{align}

In the above, (\ref{eq:set-cover-1}) follows from (\ref{eq:universe}) and 
(\ref{eq:p-max}). (\ref{eq:set-cover-2}) follows from the following argument.
Note that if $\ell <n/2$ then $\binom{n}{\ell}\ge \binom{n}{\ell'}$. 
If $\ell'>n/2$, then $\binom{n}{\ell}\ge\binom{n}{\ell}/n^{\ell-\ell'}$.
Finally if $\ell'<n/2$ and $\ell\ge n/2$, then $\binom{n}{\ell}\ge\binom{n}{\ell'}$ if $|n/2-\ell|<|n/2-\ell'|$ otherwise $\binom{n}{\ell}\ge\binom{n}{\ell'}/n^{|n/2-\ell|-|n/2-\ell'|}$. Thus, in all cases, $\binom{n}{\ell}\ge\binom{n}{\ell'}/n^{\ell-\ell'}\ge \binom{n}{\ell'}/n^{k-1}$, where the last inequality
follows from the fact that $\ell-\ell'\le k-1$.
(\ref{eq:set-cover-3}) follows from Lemma~\ref{lem:binom-ratio}. Finally (\ref{eq:set-cover-4}) follows from the fact that $\ell'\ge \ell-k+1$
and $s\le n$.

We are done except for the following lemma:

\begin{lemma}
\label{lem:binom-ratio}
Let $a\le n$ and $b\le n-a$ be integers. Then
\[\frac{\binom{n}{b}}{\binom{n-a}{b}}\ge e^{ab/n}\cdot\sqrt{\frac{8}{27(n+1)}}.\]
\end{lemma}
\begin{proof}
Stirling's approximation can be used to obtain the following bound for $n\ge 1$,
\[\sqrt{2\pi n}\left(\frac{n}{e}\right)^n\le n!\le \sqrt{3\pi n}\left(\frac{n}{e}\right)^n.\]
In particular, this implies that for any $x\le y$,
\[\frac{1}{3}\sqrt{\frac{2y}{\pi x(y-x)}}\cdot \frac{y^y}{x^x(y-x)^{y-x}}\le
\binom{y}{x}\le \frac{1}{2}\sqrt{\frac{3y}{\pi x(y-x)}}\cdot \frac{y^y}{x^x(y-x)^{y-x}}.
\]
Using the bound above, we get
\begin{equation}
\label{eq:ratio-binom-1}
\frac{\binom{n}{b}}{\binom{n-a}{b}}\ge f(a,b,n)\cdot \frac{n^n (n-a-b)^{n-a-b}}{(n-b)^{n-b}(n-a)^{n-a}},
\end{equation}
where
\begin{equation}
\label{eq:ratio-binom-2}
f(a,b,n)= \sqrt{\frac{8}{27}\cdot\frac{n(n-a-b)}{(n-b)(n-a)}}
\ge \sqrt{\frac{8}{27\left(1+\frac{ab}{n(n-a-b)}\right)}}\ge \sqrt{\frac{8}{27(1+n)}},
\end{equation}
where the last inequality used the facts that $ab\le n^2$ and $n-a-b\ge 1$.

Now consider the following sequence of relationships
\begin{align}
\frac{n^n (n-a-b)^{n-a-b}}{(n-b)^{n-b}(n-a)^{n-a}}&=\left(\frac{n}{n-a}\right)^b\left(\frac{n}{n-b}\right)^a\left(\frac{n(n-a-b)}{(n-a)(n-b)}\right)^{n-a-b}\nonumber\\
&=\frac{1}{\left(1-\frac{a}{n}\right)^b}\cdot 
\frac{1}{\left(1-\frac{b}{n}\right)^a}\cdot
\frac{1}{\left(1+\frac{ab}{n(n-a-b)}\right)^{n-a-b}}\nonumber\\
&\ge \frac{1}{e^{-ab/n}}\cdot\frac{1}{e^{-ab/n}}\cdot \frac{1}{e^{ab/n}}\nonumber\\
\label{eq:ratio-binom-3}
&=e^{ab/n},
\end{align}
where the inequality follows from the following two facts (for $x,y>0$):
\[\left(1+\frac{x}{y}\right)^y \le e^x\text{ and } (1-x)^y\le e^{-xy}.\]

(\ref{eq:ratio-binom-1}), (\ref{eq:ratio-binom-2}) and (\ref{eq:ratio-binom-3})
complete the proof.
\end{proof}

Note that our reduction maps one tolerant testing problem instance (where say
we want to distinguish between at most $e$ error vs. at least $2e$ errors) to
$O(e\log{n})$ many instances of error detection. Next we show that this is
essentially unavoidable even if we use $k^{++}$ queries for constant $k$.
The following result follows from the results in~\cite{k-plus}:

\begin{theorem}
\label{thm:query-lb}
Let $0\le \ell < t\le n$ and $k\ge 1$ be integers. 
Then any {\em adaptive}, deterministic $k^{++}$ algorithm for the $(\ell,t,n)$-threshold
problem  makes $\Omega(\ell/k)$ queries.
\end{theorem}

\subsection{Proof of Theorem~\ref{thm:query-lb}}
The proof will be by an adversarial argument to show that if 
$r<\ell/k$ $k^{++}$ queries are made then there exist two inputs $\vx$ and $\vy$ on which the
answers to the queries will be the same, yet $\wt(\vx)\le \ell$ and $\wt(\vy)\ge t$. Note
that the existence of such a pair of inputs will complete our proof.

We will think of the adversary as maintaining a set of positions $U(i)$ after the $i$th step.
The invariance that the adversary will maintain is that $U(i-1)\subseteq U(i)$ and
more importantly, that any input $\vx$ such that $\vx_{U(r)}=\myvec{1}$ will be consistent with
answers to the queries. Finally, it is also the case that $|U(i)|\le ki$. Note that
if we can come up with a way to construct these subsets $\{U(i)\}_{i=1}^r$,
then the proof will be done (consider the inputs $\myvec{1}$ and $\va$ such that $\va_{U(r)}=\myvec{1}$
and $\va_{[n]\setminus U(r)}=\myvec{0}$).

To complete the proof, we will show how the adversary can construct the set $U(i)$.
Given the $i$th $k^{++}$ query $S\subseteq[n]$, the adversary constructs $U(i)$ as follows:
Let $S'=S\setminus U(i-1)$. If $|S'|\le k$, then let $U(i)=U(i-1)\cup S$. Otherwise pick
an arbitrary subset $T\subseteq S'$ such that $|T'|=k$ and define $U(i)=U(i-1)\cup T$.
In both cases, the adversary answers the query as follows: If $|S|\ge k$, return an
answer of $k^{++}$, otherwise report that the substring indexed by $S$ is the
all ones vector. It is easy to check that $U(i)$ satisfies all the required 
properties.

\section{Randomness Efficient Construction of List Disjunct Matrices}
\label{app:list-disj}

In this section, we show that $(e,e)$-list disjunct matrices can be 
constructed with $t=O(e\log{n})$ rows (each with support at least $n/(2e)$)
with $O(e\log^2{n})$
random bits.

For this, we will need Nisan's PRG for space bounded computation~\cite{nisan-prg}.
Nisan's result states that there exists a function $G:\{0,1\}^T\rightarrow\{0,1\}^R$ such
that any Finite State Machine that uses
$O(S)$ space and $R$ random bits, cannot distinguish between truly unbiased
random
$R$ bits and the bits $G(x)$ (for $x$ chosen randomly from $\{0,1\}^T$)
for $T=O(S\log{R})$ with probability more than $2^{-O(S)}$. Further, 
any bit of $G(x)$
can be computed,
given
the $T$ random bits $x$ (and $O(S)$ extra storage).

We first use the probabilistic method to show that the required object exists
with high probability. Then we show that the proof can be implemented in
low space and use Nisan's PRG to complete the proof.

Let $t=c\cdot e\log{n}$, where $c$ is some large enough constant so that
all calculations go through. Also let $\a\ge 1$ be a large enough constant to
be determined later.
We will also assume that $t\le n$ so that $n/e\ge c\log{n}$.
Let $M$ be a random $t\times n$ matrix, where each entry is one independently with
probability $1/e$. Now to prove that $M$ has the required property, we show
that it satisfies the following two properties with high probability:
\begin{enumerate}
\item [(a)] Every row of $M$ has Hamming weight at least $\frac{n}{2e}$.
\item[(b)] For any two disjoint subsets $S,T\subseteq [n]$ such that
$|S|=|T|=e$, there is at least one row such that at least one column in
$T$ has a one in it while all the columns in $S$ have a zero in it.
\end{enumerate}

We begin with (a). Note that in expectation any row has $n/e$ ones in it. Thus,
by Chernoff bound, the probability that any row has Hamming weight at most
$n/(2e)$ is upper bounded by
\[\exp\left(-\frac{n}{12e}\right)\le \exp\left(-\frac{c}{12}\cdot\log{n}\right)
\le n^{-2\alpha},\]
where the last inequality follows for $c\ge 24\a$ and the first inequality
follows from the assumption that $t\le n$.

Next, we move to (b). Fix a row $j\in [t]$. Now the probability that
$\vee_{i\in T} M_{j,i}=1$ and $\vee_{i\in S} M_{j,1}=0$ is exactly
\[\left(1-\left(1-\frac{1}{e}\right)^e\right)\left(1-\frac{1}{e}\right)^e\ge \frac{1}{8},\]
where the last inequality follows for $e\ge 2$. Thus, the probability that there does
not exist a row as desired in part (b) is upper bounded by 
\[\left(\frac{7}{8}\right)^{ce\log{n}}\le n^{-(2+\a)e},\]
where the last inequality follows for $c\ge 20(2+\a)$. Thus, by the union bound,
part (a) does not hold with probability at most $n^{-2\a}$ (for $e\ge 2$).

Thus, $M$ does not have the desired property with probability at most $n^{-\a}$
(for $n\ge 2$).

Next, we estimate the space required to implement the proof above, i.e.
given $R=nt$ bits of the entries in $M$, we need to figure out how much space
is needed to verify whether $M$ has the required property or not. For
part (a), we need $O(\log{t}+\log{n})$ space to keep track of the row and $O(\log(n/e))$
bits to check if the row has Hamming weight at least $n/(2e)$. So we can
implement part (a) with $O(\log{n})$ space. For part (b), we need $O(e\log{n})$
space to keep track of the subsets $S$ and $T$. For given $S$ and $T$, we need
$O(\log{t})$ space to keep track of the rows and $O(\log{n})$ space to verify
if it is the row that ``takes care" of $S$ and $T$. Thus, for part (b) we
need $O(e\log{n})$ space.

Thus, overall we have $S=O(e\log{n})$. 
We are almost done, except for one small catch:
Nisan's PRG deals with unbiased bits but we need random bits that are biased.
However, we can obtain a random bit that is one with probability $1/e$
from $O(\log{e})$ unbiased bits (by declaring the final bit to be one if and
only if all the unbiased bits are $1$). Thus, we can convert the proof above to
use $R'=O(\log{e}\cdot R)$ unbiased random bits. Further, this conversion needs
an extra $O(\log\log{e}+\log{R})$ space, which implies that the
total space used is  $S'=O(e\log{n})$.

Thus, by Nisan's PRG we would be
done with $O(S'\log{R'})=O(e\log^2{n})$ random bits. Using Nisan's PRG will
increase the error probability to $n^{-\a}+2^{-O(S)}$, which can be
made to be polynomially small by picking $\a$ appropriately.

\subsection*{Acknowledgments} We thank Venkat Guruswami, Steve Li and Ram Swaminathan for helpful discussions. Thanks to Chris Umans for pointing out~\cite{limited-indep-chernoff} to us.

\bibliographystyle{abbrv}
\bibliography{code-test-ds}

\begin{thebibliography}{10}

\bibitem{adleman-lenstra}
L.~M. Adleman and H.~W. Lenstra.
\newblock Finding irreducible polynomials over finite fields.
\newblock In {\em STOC '86: Proceedings of the eighteenth annual ACM symposium
  on Theory of computing}, pages 350--355, New York, NY, USA, 1986. ACM.

\bibitem{rs-fast}
M.~Alekhnovich.
\newblock Linear diophantine equations over polynomials and soft decoding of
  reed-solomon codes.
\newblock {\em IEEE Transactions on Information Theory}, 51(7):2257--2265,
  2005.

\bibitem{ALMSS98}
S.~Arora, C.~Lund, R.~Motwani, M.~Sudan, and M.~Szegedy.
\newblock Proof verification and the hardness of approximation problems.
\newblock {\em J. ACM}, 45(3):501--555, 1998.

\bibitem{AS98}
S.~Arora and S.~Safra.
\newblock Probabilistic checking of proofs: A new characterization of {NP}.
\newblock {\em J. ACM}, 45(1):70--122, 1998.

\bibitem{k-plus}
J.~Aspnes, E.~Blais, M.~Demirbas, R.~O'Donnell, A.~Rudra, and S.~Uurtamo.
\newblock k+ decision trees, 2010.
\newblock Manuscript.

\bibitem{limited-indep-chernoff}
M.~Bellare and J.~Rompel.
\newblock Randomness-efficient oblivious sampling.
\newblock In {\em Proceedings of the 35th Annual Symposium on Foundations of
  Computer Science (FOCS)}, pages 276--287, 1994.

\bibitem{redundant}
E.~Ben-Sasson, V.~Guruswami, T.~Kaufman, M.~Sudan, and M.~Viderman.
\newblock Locally testable codes require redundant testers.
\newblock In {\em IEEE Conference on Computational Complexity}, pages 52--61,
  2009.

\bibitem{cnf-hard}
E.~Ben-Sasson, P.~Harsha, and S.~Raskhodnikova.
\newblock Some 3cnf properties are hard to test.
\newblock {\em SIAM J. Comput.}, 35(1):1--21, 2005.

\bibitem{gao-book}
I.~F. Blake, S.~Gao, A.~J.~M. (Editor), R.~C. Mulin, S.~A. Vanstone, and
  T.~Yaghoobian.
\newblock {\em Applications of Finite Fields}.
\newblock Kluwer Academic Publishers, 1993.

\bibitem{ecc-mem}
C.~L. Chen and M.~Y. Hsiao.
\newblock Error-correcting codes for semiconductor memory applications: A
  state-of-the-art review.
\newblock {\em IBM Journal of Research and Development}, 28(2):124--134, 1984.

\bibitem{chen94raid}
P.~M. Chen, E.~K. Lee, G.~A. Gibson, R.~H. Katz, and D.~A. Patterson.
\newblock {RAID}: High-performance, reliable secondary storage.
\newblock {\em ACM Computing Surveys}, 26(2):145--185, 1994.

\bibitem{dinur}
I.~Dinur.
\newblock The {PCP} theorem by gap amplification.
\newblock {\em J. ACM}, 54(3):12, 2007.

\bibitem{test-book}
D.-Z. Du and F.~K. Hwang.
\newblock {\em Combinatorial Group Testing and its Applications}.
\newblock World Scientific, 2000.

\bibitem{disk-cacm}
J.~Elerath.
\newblock Hard-disk drives: The good, the bad, and the ugly.
\newblock {\em Communications of the ACM}, 52(6):38--45, 2009.

\bibitem{furedi}
Z.~F{\"u}redi.
\newblock On $r$-cover-free families.
\newblock {\em J. Comb. Theory, Ser. A}, 73(1):172--173, 1996.

\bibitem{GR-tolerant}
V.~Guruswami and A.~Rudra.
\newblock Tolerant locally testable codes.
\newblock In {\em Proceedings of the 9th InternationalWorkshop on Randomization
  and Computation (RANDOM)}, pages 306--317, 2005.

\bibitem{GR-capacity}
V.~Guruswami and A.~Rudra.
\newblock Explicit codes achieving list decoding capacity: {E}rror-correction
  up to the {S}ingleton bound.
\newblock {\em IEEE Transactions on Information Theory}, 54(1):135--150,
  January 2008.

\bibitem{list-disj}
P.~Indyk, H.~Q. Ngo, and A.~Rudra.
\newblock Efficiently decodable non-adaptive group testing.
\newblock In {\em Proceedings of the 20th Annual ACM-SIAM Symposium on Discrete
  Algorithms (SODA)}, pages 1126--1142, 2010.

\bibitem{cc-book}
E.~Kushilevitz and N.~Nisan.
\newblock {\em Communication Complexity}.
\newblock Cambridge University Press, 1997.

\bibitem{nisan-prg}
N.~Nisan.
\newblock Pseudorandom generators for space-bounded computation.
\newblock {\em Combinatorica}, 12(4):449--461, 1992.

\bibitem{two-theorems}
A.~Rudra and S.~Uurtamo.
\newblock Two theorems in list decoding.
\newblock {\em ECCC Technical Report TR10-007}, 2010.

\bibitem{shoup-book}
V.~Shoup.
\newblock {\em A Computational Introduction to Number Theory and Algebra}.
\newblock Cambridge University Press, 2008.
\newblock 2nd Edition.

\bibitem{spielman}
D.~Spielman.
\newblock Linear-time encodable and decodable error-correcting codes.
\newblock {\em IEEE Transactions on Information Theory}, 42(6):1723--1732,
  1996.

\bibitem{sudan-notes}
M.~Sudan.
\newblock Algorithmic introduction to coding theory, 2001.
\newblock Lecture Notes available at {\tt
  http://people.csail.mit.edu/madhu/FT01/}.

\end{thebibliography}

\newpage

\appendix

\section{Upper bounding a sum}
\label{app:sum-ub}

We begin with the sum
\[\sum_{j=0}^{k-1} \binom{i}{j}\binom{t}{j}\binom{n-i}{t-j}.\]

Let $s, i, t$ and $n$ be such that $s\le i\le n$ and $t\le n$. The sum
above is then upper bounded by
\[\sum_{j=0}^{k-1} \binom{n}{j}\binom{n}{j}\binom{n-s}{t-j}.\]
Now define $j^*$ to be the minimum $0\le j\le k-1$ such that 
$t-j^*\le n-s$ (if no such $j^*$ exists then the sum is $0$). Now upper bounding $\binom{n}{j}\le n^{k-1}$ for
$j\le k-1$, we can again upper bound the sum above by
\[n^{2(k-1)}\sum_{j=j^*}^{k-1}\binom{n-s}{t-j}.\]
From the bound that $\binom{a}{b}\le a\binom{a}{b-1}$, we get that
$\binom{n-s}{t-j}\le (n-s)^{j-j^*}\binom{n-s}{t-j^*}\le n^{k-1}\binom{n-s}{t-j^*}$. This along with the bound above implies that
\[\sum_{j=0}^{k-1} \binom{i}{j}\binom{t}{j}\binom{n-i}{t-j}\le kn^{3(k-1)}\binom{n-s}{t'},\]
where $t'=t-j^*$, as desired.
Similarly one can show that
\[\sum_{j=0}^{k-1} \binom{i}{j}\binom{\ell}{j}\binom{n-i}{\ell-j}\le kn^{3(k-1)}\binom{n-s}{\ell'},\]
where $\ell'=\max_{0\le j\le k-1} \{\ell-j|\ell-j\le n-s\}$.

\end{document}